\documentclass[a4paper,11pt]{article}
\usepackage[utf8]{inputenc}
\pdfoutput=1 
\usepackage{fullpage}
\usepackage{natbib}

\usepackage{microtype}
\usepackage{hyperref}
\usepackage[svgnames]{xcolor}
\hypersetup{colorlinks={true},linkcolor={DarkBlue},citecolor={DarkGreen}}
\usepackage{amsmath,amsthm,amssymb,algorithm,algpseudocode,graphicx,tikz,caption,subcaption}
\definecolor{light-gray}{gray}{0.7}

\algrenewcommand\algorithmicindent{0.6em}
\newtheorem{theorem}{Theorem}[section]
\newtheorem{proposition}[theorem]{Proposition}
\newtheorem{lemma}[theorem]{Lemma}
\algnotext{EndFor}
\algnotext{EndIf}
\algnotext{EndWhile}
\algnotext{EndProcedure}

\theoremstyle{definition}
\newtheorem{definition}[theorem]{Definition}

\newtheorem{claim}{Claim}
\newtheorem*{claim*}{Claim}

\newtheorem*{remark*}{Remark}

\DeclareMathOperator*{\argmin}{arg\,min}

\title{Contiguous Cake Cutting:\\ Hardness Results and Approximation Algorithms\thanks{A preliminary version of this paper appeared in Proceedings of the 34th AAAI Conference on Artificial Intelligence (AAAI 2020).  
}}
\author{Paul W. Goldberg \and Alexandros Hollender \and Warut Suksompong \and\\
Department of Computer Science, University of Oxford\\
{\tt \{paul.goldberg,alexandros.hollender,warut.suksompong\}@cs.ox.ac.uk}}

\date{}

\begin{document}

\maketitle

\begin{abstract}
We study the fair allocation of a cake, which serves as a metaphor for a divisible resource, under the requirement that each agent should receive a contiguous piece of the cake.
While it is known that no finite envy-free algorithm exists in this setting, we exhibit efficient algorithms that produce allocations with low envy among the agents.
We then establish NP-hardness results for various decision problems on the existence of envy-free allocations, such as when we fix the ordering of the agents or constrain the positions of certain cuts.
In addition, we consider a discretized setting where indivisible items lie on a line and show a number of hardness results extending and strengthening those from prior work. Finally, we investigate connections between approximate and exact envy-freeness, as well as between continuous and discrete cake cutting.
\end{abstract}

\section{Introduction}

We consider the classical \emph{cake cutting} problem, where we wish to divide a cake among a set of agents with different preferences over different parts of the cake.
The cake serves as a metaphor for any divisible resource such as time or land, and our aim is to perform the division in a \emph{fair} manner.
This problem has a long and storied history that dates back over 70 years \citep{BramsTa96,RobertsonWe98,Procaccia16} and has received considerable attention in the past decade \citep{CaragiannisLaPr11,BeiChHu12,AumannDoHa13,BalkanskiBrKu14,BranzeiMi15,AlijaniFaGh17,MenonLa17,BeiHuSu18,Segalhalevi18}.

In order to reason about fairness, we need to specify when a division is considered to be fair.
One of the most commonly used definitions is \emph{envy-freeness}, which means that no agent envies another with respect to the division. 
In other words, among the pieces in the division, every agent receives their first choice.
An early result by \cite{DubinsSp61} shows that an envy-free allocation always exists for arbitrary valuations of the agents.
However, as \cite{Stromquist80} noted, this result depends on a liberal definition of what constitutes a piece of cake, and an agent ``who hopes only for a modest interval of cake may be presented instead with a countable union of crumbs.''

In light of this concern, \cite{Stromquist80} strengthened the result of Dubins and Spanier by showing that it is possible to guarantee an envy-free allocation in which every agent receives a \emph{contiguous} piece of the cake.
Stromquist's result, together with its topological proof, is widely regarded as a cornerstone of the cake-cutting literature.
Nevertheless, since the result focuses only on the existence of a contiguous envy-free allocation, it leaves open the question of how to compute such an allocation.
Almost 30 years later, Stromquist himself addressed this question and showed that under the Robertson-Webb model, where an algorithm is allowed to discover the agents' valuations through \emph{cut} and \emph{evaluate} queries,
no finite algorithm can compute a contiguous envy-free allocation when there are at least three agents \citep{Stromquist08}.\footnote{For two agents, the well-known \emph{cut-and-choose} protocol, which lets the first agent cut the cake into two equal pieces and lets the second agent choose the piece that she prefers, computes a contiguous envy-free allocation.}

Although Stromquist's later result rules out the possibility of computing contiguous envy-free allocations in general, several important questions still remain.
For instance, can we compute a contiguous allocation with low envy between the agents, and if so, how efficiently? 
How does the answer change if we know that the agents' valuations belong to a restricted class?
What happens if we add extra requirements on the allocation, such as fixing a desired ordering of the agents or constraining the positions of certain cuts?
The goal of this paper is to shed light on the complexity of contiguous cake cutting by addressing these questions.

\subsection{Our Contributions}

First, in Section~\ref{sec:approx-alg} we present two algorithms that compute an allocation with low envy in polynomial time. 
As is standard in the cake-cutting literature, we represent the cake by the interval $[0,1]$ and normalize the agents' valuations so that each agent has value $1$ for the entire interval.
Our first algorithm works for general valuations under the Robertson-Webb model and produces a contiguous allocation in which any agent has envy no more than $1/3$ towards any other agent.
On the other hand, our second algorithm is specific to valuations where each agent only desires a single subinterval and has a uniform value over that interval---for such valuations, the algorithm produces a contiguous allocation with a lower envy of at most $1/4$.

Next, in Section~\ref{sec:hardness-cake} we consider variants of the cake-cutting problem where we impose constraints on the desired allocation.
We show that for several natural variants, the decision problem of whether there exists a contiguous envy-free allocation satisfying the corresponding constraints is NP-hard.
In particular, this holds for the variants where (i) a certain agent must be allocated the leftmost piece; (ii) the ordering of the agents is fixed; and (iii) one of the cuts must fall at a given position.
Fixing the ordering of the agents is relevant when there is a temporal ordering in which the agents must be served, e.g., due to notions of seniority or the ease of switching from one agent to another in the service.
Likewise, fixing a cut point is applicable when we divide a parcel of land and there is a road crossing the parcel, so we cannot allocate a piece that lies on both sides of the road.
Moreover, our construction serves as a general framework that can be used to obtain hardness results for other related variants.

In Section~\ref{sec:hardness-indiv} we investigate a discrete analog of cake cutting, where there are indivisible items on a line and each agent is to be allocated a contiguous block of items.
The discrete setting can be viewed as a type of restriction for the continuous setting, where cuts must be placed between discrete items.
In addition to envy-freeness, we work with two other well-studied fairness notions: \emph{proportionality} and \emph{equitability}.\footnote{See the definitions in Section~\ref{sec:hardness-indiv}.}
Using a single reduction, we show that deciding whether there exists a contiguous fair allocation is NP-hard for each of the three fairness notions as well as any combination of them; our result holds even when all agents have \emph{binary} valuations\footnote{That is, the valuations are additive and each agent values each item either $0$ or $1$.} and moreover value the same number of items.
This significantly strengthens a result of \citet{BouveretCeEl17}, who established the hardness for proportionality and envy-freeness using additive but non-binary valuations. Moreover, we show that even if we consider approximate envy-freeness instead of exact, the decision problem remains NP-hard for binary valuations.
We also prove that when the valuations are binary and every agent values a contiguous block of items, deciding whether a contiguous proportional allocation exists is NP-hard.

Finally, in Section~\ref{sec:connections} we present a number of connections between approximate and exact envy-freeness, as well as between the continuous and discrete settings. First, we prove that for piecewise constant valuations, finding an approximately envy-free allocation is as hard as finding an exactly envy-free allocation. Then, we reveal some relationships between continuous and discrete cake cutting---among other things, we show that a special case of the continuous problem for piecewise constant valuations is computationally equivalent to a discrete cake-cutting problem where every item is positively valued by at most one agent. This means that any algorithm or hardness result for one problem will immediately transfer over to the other.

\subsection{Further Related Work}

Since the seminal work of \cite{Stromquist80,Stromquist08}, a number of researchers have studied cake cutting in view of the contiguity condition.
\cite{Su99} proved the existence of contiguous envy-free allocations using Sperner's lemma arguments.
\cite{DengQiSa12} showed that contiguous envy-free cake cutting is PPAD-complete; however, the result requires non-standard (e.g., non-additive, non-monotone) preference functions.
\cite{AumannDoHa13} considered the problem of maximizing social welfare with contiguous pieces, while \cite{BeiChHu12} tackled the same problem with the added requirement of proportionality.
\cite{CechlarovaPi12} and \cite{CechlarovaDoPi13} examined the existence and computation of contiguous equitable allocations---among other things, they showed that such an allocation is guaranteed to exist even if we fix the ordering of the agents.
\cite{AumannDo15} analyzed the trade-off between fairness and social welfare in contiguous cake cutting.
\cite{SegalhaleviHaAu16} circumvented \cite{Stromquist80}'s impossibility result by presenting bounded-time contiguous envy-free algorithms that may not allocate the entire cake but guarantee every agent a certain positive fraction of their value.\footnote{Without this guarantee, it would be much easier to find a contiguous envy-free allocation---just don't allocate any of the cake!}

The contiguity requirement has also been considered in the context of indivisible items.
\cite{MarencoTe14} proved that if the items lie on a line and every item is positively valued by at most one agent, a contiguous envy-free allocation is guaranteed to exist.
When each item can yield positive value to any number of agents, \cite{BarreraNyRu15}, \cite{BiloCaFl19}, and \cite{Suksompong19} showed that various relaxations  of envy-freeness can be fulfilled. 
In addition, contiguity has been studied in the more general model where the items lie on an arbitrary graph \citep{BouveretCeEl17,IgarashiPe19,BeiIgLu19}.
Like us, \citet{IgarashiPe19} also showed hardness results for binary valuations.

Recently, \cite{ArunachaleswaranBaKu19} developed an efficient algorithm that computes a contiguous cake division with multiplicatively bounded envy---in particular, each agent's envy is bounded by a multiplicative factor of $3$.
We remark that our approximation algorithms are incomparable to their result.
On the one hand, their algorithm may return an allocation wherein an agent has value $1/4$ for her own piece and $3/4$ for another agent's piece---this corresponds to an additive envy of $1/2$.
On the other hand, our algorithms may leave some agents empty-handed, leading to unbounded multiplicative envy.
We also note that additive envy is the more commonly considered form of approximation, both for cake cutting \citep{DengQiSa12,BranzeiNi17,BranzeiNi19} and for indivisible items \citep{LiptonMaMo04,CaragiannisKuMo16}. 

\section{Preliminaries}

For any positive integer $n$, let $[n] = \{1,2, \dots, n\}$.
In our cake cutting setting, we consider the cake as the interval $[0,1]$. 
There are $n$ agents whose preferences over the cake are represented by valuation functions $v_1, \dots, v_n$. 
Assume that these valuation functions are non-negative density functions over $[0,1]$. 
We abuse notation and let $v_i(a,b) = v_i([a,b]) = \int_a^b v_i(x) dx$ for $0 \leq a \leq b \leq 1$. 
It follows that the valuations are non-negative, additive, and non-atomic (i.e., $v_i(a,a) = 0$).
We assume further that the valuations are normalized so that $v_i(0,1)=1$ for every $i\in[n]$. 

A \emph{contiguous allocation} of the cake is a partition of $[0,1]$ into $n$ (possibly empty) intervals, along with an assignment of each interval to an agent, so that every agent gets exactly one interval. 
Note that this means that we cut the cake using $n-1$ cuts.
Formally, a contiguous allocation is represented by the cut positions $0 \leq x_1 \leq x_2 \leq \dots \leq x_{n-1} \leq 1$ and a permutation $\pi: [n] \to [n]$ that assigns the intervals to the agents so that agent $i$ receives the interval $[x_{\pi(i)-1},x_{\pi(i)}]$, where we define $x_0=0$ and $x_n=1$ for convenience.


We are interested in finding a contiguous allocation that is \emph{envy-free}, i.e., no agent thinks that another agent gets a better interval. 
Formally, the contiguous allocation $(x, \pi)$ is envy-free if for all $i,j \in [n]$, we have $v_i(x_{\pi(i)-1},x_{\pi(i)}) \geq v_i(x_{j-1},x_j)$.
In some cases we will be interested in finding a contiguous allocation that is only \emph{approximately} envy-free. For $\varepsilon \in [0,1]$, the contiguous allocation $(x,\pi)$ is \emph{$\varepsilon$-envy-free} if for all $i,j \in [n]$, we have $v_i(x_{\pi(i)-1},x_{\pi(i)}) \geq v_i(x_{j-1},x_j) - \varepsilon$. In other words, any agent has envy that is at most a fraction $\varepsilon$ of her value for the whole cake.

A typical way for an algorithm to access the valuation functions is through queries in the \emph{Robertson-Webb model}: the algorithm can make \emph{evaluate} queries---where it specifies $x,y$ and asks agent $i$ to return the value $v_i(x,y)$---and \emph{cut} queries---where it specifies $x,\alpha$ and asks agent $i$ to return the leftmost point $y$ such that $v_i(x,y)=\alpha$ (or say that no such $y$ exists).
A more restrictive class of valuations is that of \emph{piecewise constant} valuations.
A piecewise constant valuation function is defined by a piecewise constant density function on $[0,1]$, i.e., a step function.
This class of valuations can be explicitly represented as part of the input.
A subclass of piecewise constant valuations is the class of \emph{piecewise uniform} valuations, where the density function of agent $i$ is either some fixed rational constant $c_i$ or $0$.

\section{Approximation Algorithms}
\label{sec:approx-alg}

In this section, we present two algorithms for approximate envy-free cake cutting.
Algorithm~\ref{alg:approx-EF-general} works for arbitrary valuations and returns a $1/3$-envy-free allocation.
On the other hand, Algorithm~\ref{alg:approx-EF-single-interval} can be used for piecewise uniform valuations with a single value-block and outputs a $1/4$-envy-free allocation.
Note that such valuations are relevant, for example, when the agents are dividing machine
processing time: each agent has a release date and a deadline for her job, so she would like to maximize
the processing time she obtains after the release date and before the deadline.

\begin{algorithm}
\caption{$1/3$-Envy-Free Algorithm for Arbitrary Valuations}\label{alg:approx-EF-general}
\begin{algorithmic}[1]
\Procedure{ApproximateEFArbitrary}{}
\State $\ell\leftarrow 0$, $N\leftarrow [n]$
\For{$i\in N$}
\State $M_i\leftarrow\emptyset$
\EndFor
\While{some agent in $N$ values $[\ell,1]$ at least $1/3$} \label{line:query1}
\For{$i\in N$}
\If{$v_i(\ell,1)\geq 1/3$} \label{line:query2}
\State $r_i\leftarrow $ leftmost point such that $v_i(\ell,r_i)=1/3$ \label{line:query3}
\Else
\State $r_i\leftarrow 1$
\EndIf
\EndFor
\State $j\leftarrow \argmin_{i\in N}r_i$, $r\leftarrow\min_{i\in N}r_i$
\State $M_j\leftarrow [\ell,r]$
\State $\ell\leftarrow r$, $N\leftarrow N\backslash\{j\}$
\EndWhile
\If{$N\neq\emptyset$} \label{line:notempty}
\State $j\leftarrow $ arbitrary agent in $N$
\State $M_j\leftarrow [\ell,1]$
\Else
\State $j\leftarrow $ last agent removed from $N$
\State $M_j\leftarrow M_j\cup [\ell,1]$ \label{line:addpiece}
\EndIf
\State \Return $(M_1,\dots,M_n)$
\EndProcedure
\end{algorithmic}
\end{algorithm}

While Algorithm~\ref{alg:approx-EF-general} can be implemented for general valuations under the Robertson-Webb model, it also allows a simple interpretation as a moving-knife algorithm.
In this interpretation, the algorithm works by moving a knife over the cake from left to right.
Whenever the current piece has value $1/3$ to at least one remaining agent, the piece is allocated to one such agent.
If the knife reaches the right end of the cake, then the piece is allocated to an arbitrary remaining agent if there is at least one remaining agent, and to the agent who received the last piece otherwise.

\begin{theorem}
\label{thm:algo-general}
For $n$ agents with arbitrary valuations, Algorithm~\ref{alg:approx-EF-general} returns a contiguous $1/3$-envy-free allocation and runs in time polynomial in $n$ assuming that it makes queries in the Robertson-Webb model.
\end{theorem}

\begin{proof}
Every agent receives a single interval from the algorithm; the only possible exception is agent $j$ in line~\ref{line:addpiece}. 
However, since $j$ is chosen as the last agent removed from $N$, the interval $M_j$ allocated to $j$ earlier is adjacent to $[\ell,1]$, meaning that $j$ also receives a single interval.
Hence the allocation is contiguous.
Moreover, the algorithm only needs to make queries in lines \ref{line:query1}, \ref{line:query2} and \ref{line:query3}, and the number of necessary queries is clearly polynomial in $n$.
The remaining steps can be implemented in polynomial time.

We now prove that the envy of an agent $i$ towards any other agent is at most $1/3$.
If $i$ is assigned a piece in the while loop (line \ref{line:query1}), $i$ receives value at least $1/3$.
This means that $i$'s value for any other agent's piece is at most $2/3$, so $i$'s envy is no more than $1/3$.
Alternatively, after the while loop, $i$ still has not received a piece, meaning that $N\neq\emptyset$ in line~\ref{line:notempty}.
By our allocation procedure in the while loop, $i$ values any piece assigned in the while loop at most $1/3$.
Furthermore, when the algorithm enters line~\ref{line:notempty}, $i$ values the interval $[\ell,1]$ less than $1/3$.
Since $[\ell,1]$ is assigned to an agent who did not receive an interval earlier, it follows that $i$ does not envy any other agent more than $1/3$, as claimed.
\end{proof}

\begin{algorithm}
\caption{$1/4$-Envy-Free Algorithm for Uniform Single-Interval Valuations}\label{alg:approx-EF-single-interval}
\begin{algorithmic}[1]
\Statex \(\triangleright\) $R_i$ : the single interval valued by agent $i$
\Statex \(\triangleright\) $\text{mid}(i)$ : the midpoint of $R_i$
\Statex \(\triangleright\) $A_i$ : part of $R_i$ that is unallocated at the start of agent $i$'s turn
\Statex \(\triangleright\) an interval is \emph{restrained} if it is adjacent to an interval that has already been allocated
\Procedure{ApproximateEFSingleInterval}{}
\State Order the agents $1,\dots,n$ so that $|R_i| \leq |R_j|$ for all $i < j$
\For{$i=1,\dots,n$}
\If{there exists a \emph{restrained} interval $I \subseteq A_i$ with $v_i(I)=1/4$ and $\text{mid}(i) \in I$}\Comment{Case 1}\label{line:case1}
\State $M_i\leftarrow I$
\ElsIf{there exists an interval $I \subseteq A_i$ with $v_i(I)=1/4$ and $\text{mid}(i) \in I$}\Comment{Case 2}\label{line:case2}
\State $S_i \leftarrow \{j>i\mid v_i(\min(\text{mid}(i),\text{mid}(j)), \max(\text{mid}(i),\text{mid}(j))) \leq 1/4$\}
\State $k \leftarrow \min S_i$
\State $M_i\leftarrow $ an interval $I \subseteq A_i$ with $v_i(I)=1/4$, $\text{mid}(i) \in I$ and $\text{mid}(k) \in \partial I$ (i.e., an endpoint of $I$)
\ElsIf{there exist $\ell < i$ with $\text{mid}(i) \in M_\ell$ and an interval $I \subseteq A_i$ adjacent to $M_\ell$ with $v_i(I)=1/4$}\Comment{Case 3}\label{line:case3}
\State $M_i\leftarrow I$
\Else\Comment{Case 4}\label{line:case4}
\State $M_i\leftarrow $ a largest restrained interval $I \subseteq A_i$ with $v_i(I) \leq 1/4$
\EndIf
\EndFor
\If{some two intervals $M_q,M_r$ are adjacent (say, $M_q$ is to the left of $M_r$)} \label{line:finalstep}
\State Extend $M_q$ and all assigned intervals to its left as far as possible to the left.
\State Extend $M_r$ and all assigned intervals to its right as far as possible to the right.
\Else
\State Extend assigned intervals arbitrarily to cover the remaining cake.
\EndIf
\State \Return $(M_1,\dots,M_n)$
\EndProcedure
\end{algorithmic}
\end{algorithm}

Note that if we are only interested in having an algorithm that makes a polynomial number of queries, \cite{BranzeiNi17} showed that for any $\varepsilon>0$, a contiguous $\varepsilon$-envy-free allocation can be found using $O(n/\varepsilon)$ queries, which is polynomial in $n$ for constant $\varepsilon$. 
Their algorithm works by cutting the cake into pieces of size $1/\varepsilon$ and performing a brute-force search over the space of all contiguous allocations with respect to these cuts; this algorithm therefore has exponential computational complexity (even for constant $\varepsilon$).
By contrast, in the absence of the contiguity constraint, \citet[p.~323]{Procaccia16} gave a simple polynomial-time algorithm that computes an $\varepsilon$-envy-free allocation for any constant $\varepsilon$.
His algorithm also starts by cutting the cake into pieces of size $1/\varepsilon$ and then lets agents choose their favorite pieces in a round-robin manner; consequently, the resulting allocation can be highly non-contiguous.

While we do not know whether the bound $1/3$ in our approximation can be improved under the computational efficiency requirement,\footnote{For the case $n=3$, \cite{DengQiSa12} gave a fully polynomial-time approximation scheme that computes a contiguous $\varepsilon$-envy-free allocation for any $\varepsilon > 0$.} we show next that if the agents have piecewise uniform valuations and each agent only values a single interval, the envy can be reduced to $1/4$.
\cite{AlijaniFaGh17} showed that if the valuations are as described and moreover the $n$ valued intervals satisfy an ``ordering property'', meaning that no interval is a strict subinterval of another interval, then a contiguous envy-free allocation can be computed efficiently.
Nevertheless, the ordering property is a very strong assumption, and indeed reducing the envy to $1/4$ without this assumption already requires significant care in assigning the pieces.\footnote{\cite{AlijaniFaGh17} also showed that for piecewise uniform valuations where each agent only values a single interval (without the ordering property assumption), one can efficiently compute an envy-free allocation with at most $2n-1$ intervals in total. Moreover, they showed that for a \emph{constant} number of agents with piecewise constant valuations, a contiguous envy-free allocation can be computed efficiently.}

At a high level, Algorithm~\ref{alg:approx-EF-single-interval} first orders the agents from shortest to longest desired interval, breaking ties arbitrarily.
For each agent in the ordering, if an interval of value $1/4$ containing the midpoint of her valued interval (perhaps at the edge of the former interval) has not been taken, the agent takes one such interval.
Else, if an interval of value $1/4$ is available somewhere, the agent takes one such interval; here, if there are choices on both sides of the midpoint, the agent may need to be careful to pick the ``correct'' one.
Otherwise, if no interval of value $1/4$ is available, the agent takes a largest available interval.
At the end of this process, part of the cake may remain unallocated.
If some pair of assigned intervals are adjacent, pick one such pair, and allocate the remaining cake by extending pieces away from the border between this pair.
Else, extend the pieces arbitrarily to cover the remaining cake.

\begin{theorem}
\label{thm:algo-oneblock}
For $n$ agents with piecewise uniform valuations such that each agent only values a single interval, Algorithm~\ref{alg:approx-EF-single-interval} returns a contiguous $1/4$-envy-free allocation and runs in time polynomial in $n$.
\end{theorem}

\begin{proof}
One can check that Algorithm~\ref{alg:approx-EF-single-interval} assigns a single interval to every agent and can be implemented in polynomial time. 
It remains to show that the algorithm returns an allocation such that for any two agents $i,j$, agent $i$ has envy at most $1/4$ towards agent $j$. For the purpose of this proof, when we refer to an interval $M_i$, we mean the interval before it is extended in the final phase of the algorithm (the \emph{extension phase} starting at line~\ref{line:finalstep}). We denote by $M_i^+$ the corresponding extended interval that is returned by the algorithm. For any agent $i$ and any interval $I$, the $i$-value of $I$ is the value of $I$ for agent $i$, i.e., $v_i(I)$.

When agent $i$'s turn comes in the for-loop, it falls into exactly one of four possible cases: Case 1 (line~\ref{line:case1}), Case 2 (line~\ref{line:case2}), Case 3 (line~\ref{line:case3}) or Case 4 (line~\ref{line:case4}). Depending on which case applies, $M_i$ is chosen accordingly. We say that the \emph{single-direction extension} (SDE) property holds if at least one agent does not fall into Case 2. It is easy to check that if the SDE property holds, then there are at least two allocated intervals $M_q$ and $M_r$ that are adjacent before the extension phase begins, and thus every interval $M_i$ will be extended in a single direction.

It is clear that $v_i(M_j^+) \geq v_i(M_j)$ for all $i,j$. Furthermore, in all four cases it holds that agent $i$ is allocated an interval of value at most $1/4$, i.e., $v_i(M_i) \leq 1/4$ for all $i$. Since $M_i \subseteq R_i$ and because of the way the agents are ordered, it follows that
\begin{equation}\label{eqn:algo-oneblock}
v_i(M_j) \leq 1/4 \quad \text{ for all } j \leq i
\end{equation}

We now show that any agent $i$ has envy at most $1/4$ at the end of the algorithm. Namely, we prove that for any agents $i,j$ we have $v_i(M_j^+) \leq v_i(M_i^+) + 1/4$. We treat the four different cases that can occur during agent $i$'s turn.

\emph{Cases $1$ and $2$}. In both cases, $M_i$ contains $\text{mid}(i)$ and has $i$-value $1/4$. This also holds for $M_i^+ \supseteq M_i$. Since the midpoint of $R_i$ is contained in $M_i^+$, any other interval $M_j^+$ has $i$-value at most $1/2$. Thus, agent $i$ has envy at most $1/4$.

\emph{Case $3$}. In this case, we again have $v_i(M_i) = 1/4$. However, this time we have $\text{mid}(i) \in M_\ell$, which implies that $v_i(M_j^+) \leq 1/2$ for all $j \neq \ell$. Thus, it remains to show that $v_i(M_\ell^+) \leq 1/2$. Since $\ell < i$, we have $v_i(M_\ell) \leq 1/4$.
Thus, we need to show that the extension of $M_\ell$ to $M_\ell^+$ increases the $i$-value by at most $1/4$. Since $M_i$ was chosen to be adjacent to $M_\ell$, it suffices to show that there is at most $1/4$ $i$-value available on the other side of $M_\ell$.

To this end, we prove that at the start of agent $i$'s turn, $M_\ell$ cannot have at least $1/4$ of $i$-value available both on the left side and on the right side. Assume on the contrary that this is the case. Note, in particular, that $M_\ell$ is not restrained. Thus, $M_\ell$ was allocated in agent $\ell$'s turn by Case 2. We also know that $i \in S_\ell$, because $\text{mid}(i), \text{mid}(\ell) \in M_\ell$ and $v_\ell(M_\ell) = 1/4$. Now there are two cases:
\begin{itemize}
    \item If $i = \min S_\ell$, then $\text{mid}(i) \in \partial M_\ell$. But in that case, at the start of agent $i$'s turn, there exists a restrained interval $I \subseteq A_i$ with $v_i(I)=1/4$ and $\text{mid}(i) \in I$. Thus, agent $i$ would have been in Case 1 instead of 3.
    \item If $i > k = \min S_\ell$, then in agent $k$'s turn, Case 1 will apply. Indeed, $\text{mid}(k) \in \partial M_\ell$ and thus there is at least $1/4$ of $k$-value available that contains $\text{mid}(k)$ (because there is enough space for $1/4$ of $i$-value and $i>k$). But if Case 1 applies, then $M_k$ will be chosen to be adjacent to $M_\ell$ (since they both contain $\text{mid}(k)$), and $M_\ell$ will not have space available on both sides when agent $i$'s turn comes.
\end{itemize}


\emph{Case $4$}. First, suppose that $v_i(M_i) < 1/4$. This means that $M_i$ was a largest available interval in $R_i$. It follows that any agent $j > i$ can obtain an interval of $i$-value at most $v_i(M_i)$, since it is processed after $i$. For $j<i$, since agent $i$ is in Case 4, the SDE property holds. Thus, $M_j$ can be extended by at most $v_i(M_i)$, i.e., $v_i(M_j^+) \leq v_i(M_j) + v_i(M_i)$ for all $j$. With \eqref{eqn:algo-oneblock} it follows that the envy is at most $1/4$.

Now, consider the case where $v_i(M_i) = 1/4$. Any agent $j>i$ can obtain $i$-value no more than $1/2$---otherwise, agent $i$ would have fallen in Case 1 or 2. Consider any $j<i$:
\begin{itemize}
    \item if $\text{mid}(i) \in M_j$, then both on the left and right side of $M_j$ the space available has $i$-value $<1/4$ (otherwise agent $i$ would be in Case 1 or 3). Since the SDE property holds, it follows that $v_i(M_j^+) \leq v_i(M_j) + 1/4 \leq 1/2$ with \eqref{eqn:algo-oneblock}.
    \item if $\text{mid}(i) \notin M_j$, then $v_i(M_j^+) < 1/2$. Otherwise, it means that $M_j$ is extended in a single direction (SDE property) and takes over an interval of $i$-value at least $1/4$ that contains $\text{mid}(i)$. But then, agent $i$ would be in Case 1 or 2.
\end{itemize}
This completes the proof.
\end{proof}

\section{Hardness for Cake-Cutting Variants}
\label{sec:hardness-cake}

In this section, we establish hardness results for a number of decision problems on the existence of contiguous envy-free allocations.

\begin{theorem}\label{thm:cake-NP}
The following decision problems are NP-hard for contiguous cake cutting, even if we restrict the valuations to be piecewise uniform:
\begin{itemize}
    \item Does there exist an envy-free allocation in which agent $1$ obtains the leftmost piece?
    \item Does there exist an envy-free allocation in which the pieces are allocated to the $n$ agents in the order $1,2, \dots, n$?
    \item Does there exist an envy-free allocation such that there is a cut at position $x$, for $x$ given in the input?
\end{itemize}

These problems remain NP-hard if we replace envy-freeness by $\varepsilon$-envy-freeness for any sufficiently small constant $\varepsilon$.
\end{theorem}

This list is not exhaustive: additional results of the same flavor can be found in the full proof (Section~\ref{app:cake-NP}).\footnote{However, if we fix all $n-1$ cuts, the problem becomes solvable in polynomial time.
Indeed, with all the cuts fixed, the resulting pieces are also all fixed.
We can therefore construct a bipartite graph with the agents on one side and the pieces on the other side, where there is an edge between an agent and a piece exactly when receiving the piece would make the agent envy-free.
The problem of determining whether an envy-free allocation exists therefore reduces to deciding the existence of a perfect matching, which can be done in polynomial time.} 
The following proof sketch conveys the main ideas behind these results.

\begin{proof}[Proof Sketch]

In order to prove that these decision problems are NP-hard, we reduce from \textsc{3-sat}. Namely, given a \textsc{3-sat} formula, we construct a cake-cutting instance such that the answer to the decision problem is ``Yes'' if and only if the \textsc{3-sat} formula is satisfiable. A bonus of our proof is that we construct a single cake-cutting instance that works for all of the decision problems mentioned in the theorem statement and even a few more.

Let us give some insight into how this instance is constructed. Consider a \textsc{3-sat} formula $C_1 \lor \dots \lor C_m$, where the $C_i$ are clauses containing 3 literals using the variables $x_1, \dots, x_n$ and their negations. The cake-cutting instance is constructed by putting together multiple small cake-cutting instances, so-called gadgets. For every clause $C_i$ we introduce a Clause-Gadget with its three corresponding agents $C_i^{1}$, $C_i^{2}$ and $C_i^{3}$. The intuition here is that $C_i^{1}$ is associated to the first literal appearing in $C_i$, $C_i^{2}$ to the second one, and $C_i^{3}$ to the third one. For any Clause-Gadget agent $A$, we let $\ell(A)$ denote the associated literal. The valuations of these agents inside the gadget are as shown in Figure~\ref{fig:cake:clause-gadget}. We say that the gadget \emph{operates correctly} if it contains exactly two cuts and the three resulting pieces go to the three agents $C_i^{1}$, $C_i^{2}$ and $C_i^{3}$. At this point we can already make a first key observation: if the gadget operates correctly, at least one of the three agents must be \emph{sad}, i.e., obtain at most one out of its three blocks of value in this gadget.

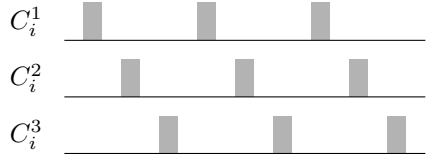
\begin{figure}
\centering
\begin{tikzpicture}[scale=1]

\fill[light-gray] (0.25,1.5) rectangle ++(0.24,0.5);
\fill[light-gray] (0.75,0.75) rectangle ++(0.24,0.5);
\fill[light-gray] (1.25,0) rectangle ++(0.24,0.5);
\fill[light-gray] (1.75,1.5) rectangle ++(0.24,0.5);
\fill[light-gray] (2.25,0.75) rectangle ++(0.24,0.5);
\fill[light-gray] (2.75,0) rectangle ++(0.24,0.5);
\fill[light-gray] (3.25,1.5) rectangle ++(0.24,0.5);
\fill[light-gray] (3.75,0.75) rectangle ++(0.24,0.5);
\fill[light-gray] (4.25,0) rectangle ++(0.24,0.5);

\draw (0,1.5) -- (4.75,1.5);
\draw (0,0.75) -- (4.75,0.75);
\draw (0,0) -- (4.75,0);

\node[font=\small] at (-0.5,1.75) {$C_i^{1}$};
\node[font=\small] at (-0.5,1) {$C_i^{2}$};
\node[font=\small] at (-0.5,0.25) {$C_i^{3}$};

\end{tikzpicture}
\caption{Clause-Gadget for clause $C_i$: the valuations of its three agents inside the gadget. Every block in this figure has value $0.24$.}
\label{fig:cake:clause-gadget}
\end{figure}

For every variable $x_j$ we introduce a Variable-Gadget with its two corresponding agents $L_j$ and $R_j$. Apart from these two agents, some Clause-Gadget agents will also have a value-block inside this gadget. In more detail, all the Clause-Gadget agents that correspond to $x_j$ or $\overline{x}_j$ will have a block of value inside the Variable-Gadget for $x_j$. Figure~\ref{fig:cake:variable-gadget} shows how the value-blocks are arranged inside the gadget. We say that the gadget operates correctly if it contains exactly one cut and the two resulting pieces go to $L_j$ and $R_j$. There is a second key observation to be made here. Assume that all gadgets operate correctly. If some agent $C_i^{k}$ with $\ell(C_i^{k}) = x_j$ (or $\overline{x}_j$) is sad, then the value-block of $C_i^{k}$ in the Variable-Gadget for $x_j$ has to contain a cut (otherwise $C_i^{k}$ would be envious). Since the Variable-Gadget contains exactly one cut, it is impossible to have agents $A$ and $B$ with $\ell(A)=x_j$ and $\ell(B)=\overline{x}_j$ that are both sad.

\begin{figure}
\centering
\begin{tikzpicture}[scale=1]

\fill[light-gray] (0.25,2.25) rectangle ++(1,0.5);
\fill[light-gray] (1.75,1.5) rectangle ++(0.28,0.5);
\fill[light-gray] (2.5,0.75) rectangle ++(0.28,0.5);
\fill[light-gray] (3.25,0) rectangle ++(1,0.5);

\draw (0,2.25) -- ++(4.5,0);
\draw (0,1.5) -- ++(4.5,0);
\draw (0,0.75)-- ++(4.5,0);
\draw (0,0) -- ++(4.5,0);

\node[font=\small] at (-0.5,2.5) {$L_j$};
\node[align=left,font=\footnotesize] at (-1,1.75) {all agents $A$: \\$\ell(A)=x_j$};
\node[align=left,font=\footnotesize] at (-1,1) {all agents $A$: \\$\ell(A)=\overline{x}_j$};
\node[font=\small] at (-0.5,0.25) {$R_j$};

\end{tikzpicture}
\caption{Variable-Gadget for variable $x_j$: the valuations of its two agents $L_j$ and $R_j$, as well as the Clause-Gadget agents that have value in this gadget. The large blocks have value $1$ each and the small blocks have value $0.28$ each.}
\label{fig:cake:variable-gadget}
\end{figure}
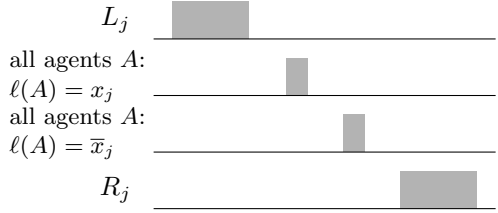

The instance is constructed by positioning the gadgets one after the other on the cake. Starting from the left and moving to the right, we first put the Clause-Gadget for $C_1$, then $C_2$, and so on until $C_m$, and then the Variable-Gadget for $x_1$, then $x_2$, and so on until $x_n$. Between adjacent gadgets we introduce a small interval without any value-blocks. We say that an envy-free allocation is \emph{nice} if all the gadgets operate correctly.

Let us now see how a nice envy-free allocation yields a satisfying assignment for the \textsc{3-sat} formula. For any agent $C_i^{k}$ that is sad, we set the corresponding literal $\ell(C_i^{k})$ to be true. This means that if $\ell(C_i^{k})=x_j$, then we set $x_j$ to be true, and if $\ell(C_i^{k})=\overline{x}_j$, then we set $x_j$ to be false. The first key observation above tells us that every Clause-Gadget has at least one sad agent. Thus, this assignment of the variables ensures that every clause is satisfied. However, we have to make sure that this assignment is consistent, i.e., we never set $x_j$ to be both true and false. This consistency is enforced by the Variable-Gadget for $x_j$ and the second key observation above.

Conversely, given a satisfying assignment for the \textsc{3-sat} formula, it is not too hard to construct a nice envy-free allocation. This proves NP-hardness for the decision problem ``Does there exist a nice envy-free allocation?''. In order to prove the result for the more natural decision problems stated in Theorem~\ref{thm:cake-NP}, the construction has to be extended with some additional work.
\end{proof}

\subsection{Proof of Theorem~\ref{thm:cake-NP}}
\label{app:cake-NP}

We provide a full proof of NP-hardness for the following decision problems:
\begin{enumerate}
    \item Does there exist an envy-free allocation in which agent $1$ gets the leftmost piece?
    \item Does there exist an envy-free allocation in which agents $1,2, \dots, k$ get the $k$ leftmost pieces, in that order? (for any constant $k \geq 1$)
    \item Does there exist an envy-free allocation in which all the agents $1,2, \dots, n$ are assigned pieces in that order from left to right?
    \item Does there exist an envy-free allocation such that there is a cut at position $x$? ($x$ given in the input)
    \item Does there exist an envy-free allocation such that the leftmost cut is at position $x$? ($x$ given in the input)
    \item Does there exist an envy-free allocation such that there are cuts at positions $x_1$, $\dots$, $x_k$? ($x_1$, $\dots$, $x_k$ given in the input, $k$ any constant)
\end{enumerate}
The problems remain NP-hard if we replace envy-freeness by $\varepsilon$-envy-freeness for any $\varepsilon \leq 0.01$.

\begin{remark*}
The list of NP-hard problems that we have provided is by no means exhaustive. The construction we provide below should be viewed as a framework for obtaining these kinds of results. Indeed, with some simple modifications, one can prove additional results of the same general flavor. In particular, one can change the constraint to ``agent $1$ gets the $k$th piece from the left ($k\geq1$ constant)'' or to ``the $k$ leftmost cuts are at positions $x_1$, $\dots$, $x_k$''.
\end{remark*}

Let $I$ be an instance of \textsc{$3$-sat} with $m$ clauses $C_1, \dots, C_m$, where each clause is made out of $3$ literals using the variables $x_1, \dots, x_n$ and their negations. Note that $m$ is polynomial in $n$ and thus we can use $n$ as the complexity parameter for the instance.
Let $\varepsilon\leq 0.01$ be arbitrary.

We will construct an instance where the cake is the interval $[0,p(n)]$ (for some polynomial $p$), instead of the usual $[0,1]$. This is just for convenience as it is easy to obtain a completely equivalent instance on $[0,1]$ in polynomial time. Indeed, it suffices to divide the position of every block by $p(n)$ and multiply its height by $p(n)$. Note that our construction also gives NP-hardness if the valuations are given in unary representation, since the positions and heights of blocks will have numerator and denominator bounded by some polynomial (even after we scale down to $[0,1]$). All the valuations we construct will be piecewise uniform, and in fact all blocks of all agents will have height $1$ (before scaling the cake to $[0,1]$), but variable length. Furthermore, value-blocks of different agents will not overlap.

\textbf{Clause-Gadget.} Consider any clause $C_i$ in the instance $I$. $C_i$ will be represented by a \emph{Clause-Gadget} in the cake cutting instance. The Clause-Gadget for $C_i$ requires an interval of length $9$ on the cake, say $[a_i,a_i+9]$, where only three specific agents are allowed to have any value. These agents are denoted by $C_i^{1}$, $C_i^{2}$ and $C_i^{3}$. The interpretation is that $C_i^{1}$ corresponds to the first literal appearing in the clause $C_i$, $C_i^{2}$ to the second one, and $C_i^{3}$ to the third one. The valuation of agent $C_i^{1}$ contains three blocks of value in the interval $[a_i,a_i+9]$: one in each of the subintervals $[a_i,a_i+1]$, $[a_i+3,a_i+4]$ and $[a_i+6,a_i+7]$. Each of these blocks has value $0.24$ (i.e., length $0.24$ and height $1$). Agents $C_i^{2}$ and $C_i^{3}$ have the same blocks as $C_i^{1}$, but shifted by $1$ and $2$ to the right respectively. The valuations of the three agents inside the Clause-Gadget are shown in Figure~\ref{fig:cake:clause-gadget}.

Note that each of the three agents has value $0.72$ inside the Clause-Gadget. The remaining $0.28$ value will be situated in a different gadget that we introduce next.

\textbf{Variable-Gadget.}  For every variable $x_j$ we introduce a \emph{Variable-Gadget} in the cake cutting instance. The Variable-Gadget for $x_j$ requires an interval of length $4$, say $[b_j,b_j+4]$, and introduces two new agents $L_j$ and $R_j$. $L_j$ has a block of value $1$ in the subinterval $[b_j,b_j+1]$, and $R_j$ has a block of value $1$ in the subinterval $[b_j+3,b_j+4]$. For every clause $C_i$ that contains $x_j$ (respectively $\overline{x}_j$) in the $\ell$th position ($\ell \in \{1,2,3\}$), the agent $C_i^{\ell}$ has a block of value $0.28$ lying at the center of the subinterval $[b_j+1,b_j+2]$ (respectively $[b_j+2,b_j+3]$). See the illustration in Figure~\ref{fig:cake:variable-gadget}.

\textbf{Instance.}  Now consider the cake-cutting instance constructed as follows: starting from the left, position all the Clause-Gadgets one after the other, leaving an interval of length $3$ after every gadget. Then, position all the Variable-Gadgets one after the other, again leaving an interval of length $3$ after every gadget. Thus, the cake is the interval $[0, 12m + 7n]$, where the first Clause-Gadget occupies the interval $[0,9]$, and the first Variable-Gadget occupies the interval $[12m,12m+4]$. There are $3m+2n$ agents so far. Note that adjacent gadgets are separated by intervals of length $3$ that we call \emph{Isolating Intervals}. There are exactly $m+n-1$ Isolating Intervals.

The $k$th Isolating Interval from the left is denoted $I_k$. The Isolating Interval $I_k=[a,a+3]$ is divided into three subintervals: $I_k[1]=[a,a+1]$, $I_k[2]=[a+1,a+2]$ and $I_k[3]=[a+2,a+3]$. Furthermore, we also add an interval of length 3 on the left end of the cake: the \emph{Initiation Interval}. We denote it by $I_0$ and it is similarly subdivided into $I_0[1]$, $I_0[2]$ and $I_0[3]$.
The cake is now represented by the interval $[0,12m+7n+3]$.

We add two new agents $S_0$ and $S_0'$. Agent $S_0$ has a block of value $1/7$ in $I_0[1]$, a block of value $2/7$ in each of $I_0[3]$, $I_1[1]$ and $I_1[3]$. Agent $S_0'$ has a block of value $1$ in $I_0[2]$. For $k \in [m+n-2]$ we define an agent $S_k$ that has a block of value $0.2$ in $I_k[2]$ and a block of value $0.4$ in each of $I_{k+1}[1]$ and $I_{k+1}[3]$. We also define an agent $S_{m+n-1}$ that has a block of value $1$ in $I_{m+n-1}[2]$. Figure~\ref{fig:cake:isolating-interval} shows the valuations of the agents in $I_0$, $I_1$ and $I_2$.
The total number of agents is $(3m+2n)+(m+n+1) = 4m+3n+1$, so there are $4m+3n$ cuts in any solution.

\begin{figure}
\centering
\begin{tikzpicture}[scale=1]

\fill[light-gray] (0.1,2.25) rectangle ++(0.15,0.5);
\fill[light-gray] (0.5,1.5) rectangle ++(1,0.5);
\fill[light-gray] (1.7,2.25) rectangle ++(0.3,0.5);

\fill[light-gray] (2.8,2.25) rectangle ++(0.3,0.5);
\fill[light-gray] (3.2,0.75) rectangle ++(0.2,0.5);
\fill[light-gray] (3.5,2.25) rectangle ++(0.3,0.5);

\fill[light-gray] (4.6,0.75) rectangle ++(0.4,0.5);
\fill[light-gray] (5.1,0) rectangle ++(0.2,0.5);
\fill[light-gray] (5.4,0.75) rectangle ++(0.4,0.5);

\draw (0,2.25) -- (2.2,2.25);\draw[dotted] (2.2,2.25) -- (2.6,2.25);\draw (2.6,2.25) -- (4,2.25);\draw[dotted] (4,2.25) -- (4.4,2.25);\draw (4.4,2.25) -- (6,2.25);
\draw (0,1.5) -- (2.2,1.5);\draw[dotted] (2.2,1.5) -- (2.6,1.5);\draw (2.6,1.5) -- (4,1.5);\draw[dotted] (4,1.5) -- (4.4,1.5);\draw (4.4,1.5) -- (6,1.5);
\draw (0,0.75) -- (2.2,0.75);\draw[dotted] (2.2,0.75) -- (2.6,0.75);\draw (2.6,0.75) -- (4,0.75);\draw[dotted] (4,0.75) -- (4.4,0.75);\draw (4.4,0.75) -- (6,0.75);
\draw (0,0) -- (2.2,0);\draw[dotted] (2.2,0) -- (2.6,0);\draw (2.6,0) -- (4,0);\draw[dotted] (4,0) -- (4.4,0);\draw (4.4,0) -- (6,0);

\node[font=\small] at (-0.5,2.5) {$S_0$};
\node[font=\small] at (-0.5,1.75) {$S_0'$};
\node[font=\small] at (-0.5,1) {$S_1$};
\node[font=\small] at (-0.5,0.25) {$S_2$};

\node at (1.1,-0.5) {$I_0$};
\node at (3.3,-0.5) {$I_1$};
\node at (5.2,-0.5) {$I_2$};

\draw[dotted] (0.35,0) -- ++(0,3);
\draw[dotted] (1.85,0) -- ++(0,3);
\draw[dotted] (2.95,0) -- ++(0,3);
\draw[dotted] (3.65,0) -- ++(0,3);
\draw[dotted] (4.8,0) -- ++(0,3);
\draw[dotted] (5.6,0) -- ++(0,3);

\end{tikzpicture}
\caption{The Initiation Interval $I_0$ and the Isolating Intervals $I_1$ and $I_2$, along with the valuations of all agents who have positive value in any of these three intervals. The vertical dotted lines indicate the position of the cuts in the envy-free allocation that we construct from a satisfying assignment.}
\label{fig:cake:isolating-interval}
\end{figure}
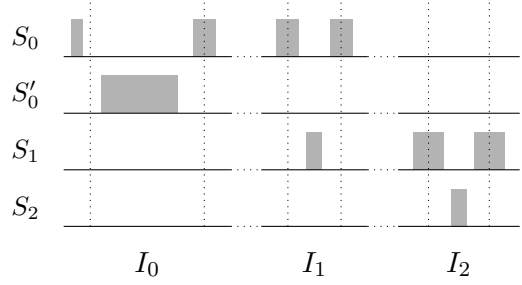

Let $\varepsilon=0.01$. 
Since an envy-free allocation always exists, the cake-cutting instance we have constructed admits an envy-free allocation (in particular also $\varepsilon$-envy-free). In order to ensure that a solution only exists if the \textsc{3-sat} formula is satisfiable, we have to add an additional constraint. An $\varepsilon$-envy-free allocation is said to satisfy the \emph{Isolation property} if together all the Clause- and Variable-Gadgets contain at most $2m+n$ cuts strictly within them.

\begin{lemma}\label{lem:isolation}
Any $\varepsilon$-envy-free allocation that satisfies the Isolation property yields a satisfying assignment for the \textsc{3-sat} formula.
\end{lemma}

\begin{proof}
Consider any $\varepsilon$-envy-free allocation. If there is at most one cut strictly inside the Clause-Gadget of $C_i$, then there is an agent $C_i^{k}$ ($k \in \{1,2,3\}$) who does not obtain any of its value from this Clause-Gadget. Thus, agent $C_i^{k}$ gets value at most $0.28$ (from its corresponding Variable-Gadget). However, since the Clause-Gadget of $C_i$ is divided into at most two parts, some agent gets at least $0.72/2=0.36$ according to agent $C_i^{k}$'s valuation, which contradicts $\varepsilon$-envy-freeness. Thus, every Clause-Gadget contains at least two cuts strictly within them.

If the Variable-Gadget for $x_j$ does not strictly contain any cut, then all of it is allocated to a single agent. Necessarily, agent $L_j$ or $R_j$ would have envy $1 > \varepsilon$. Thus, every Variable-Gadget strictly contains at least one cut.

Now consider an $\varepsilon$-envy-free allocation that also satisfies the Isolation property. Since the property permits at most $2m+n$ cuts strictly inside gadgets, we get that these lower bounds on the number of cuts inside gadgets are actually tight. Thus, there are exactly two cuts strictly inside every Clause-Gadget and exactly one cut strictly inside every Variable-Gadget.

Since there is exactly one cut strictly inside the Variable-Gadget of $x_j$, the two resulting parts must go to agents $L_j$ and $R_j$. Indeed, if one of these two agents does not get one of the two parts, then the agent would have envy at least $1/2 > \varepsilon$. Similarly, since there are exactly two cuts strictly inside the Clause-Gadget of $C_i$, the three resulting parts must go to agents $C_i^{1}$, $C_i^{2}$ and $C_i^{3}$.
Indeed, if one of these three agents does not get one of the three parts, the agent would have value $0$ (as she cannot get any value from the corresponding variable gadget) and therefore have envy at least $0.24 > \varepsilon$.

We now show how such a solution yields a satisfying assignment to the \textsc{3-sat} instance. Consider the Clause-Gadget of $C_i$. As we showed above, there are exactly two cuts strictly inside the gadget and the three resulting parts go to the agents $C_i^{1}$, $C_i^{2}$ and $C_i^{3}$. Any of these three agents who obtains at most $0.24$ of its own value is called \emph{sad}. By inspection of the construction of the Clause-Gadget it follows that at least one of the three agents must be sad. Indeed, it is easy to check that if $C_i^{1}$ is not sad, then at least one of the other two must be. The fact that any Clause-Gadget must have at least one sad agent will be used to encode the fact that any clause of the \textsc{3-sat} instance must have at least one literal set to $1$. Thus, if $C_i^{k}$ is sad, this means that we set the literal corresponding to $C_i^{k}$ to have the value $1$.

It remains to check that this is consistent, i.e., that we never set the two literals $x_j$ and $\overline{x}_j$ to both be $1$. In other words, we have to check that if some agent $C_i^{k}$ corresponding to the literal $\ell \in \{x_j,\overline{x}_j\}$ is sad, then all agents corresponding to $\overline{\ell}$ are not sad. If agent $C_i^{k}$ is sad, then it gets value at most $0.24$. Agent $C_i^{k}$ has a block of value $0.28$ in the Variable-Gadget of $x_j$. This block must contain a cut, otherwise $C_i^{k}$ would have envy at least $0.04 > \varepsilon$. But since there is a single cut inside the Variable-Gadget, the blocks of all agents corresponding to $\overline{\ell}$ are not cut. As a result, these agents cannot be sad.
\end{proof}

\begin{claim}\label{clm:value-larger-zero}
In any $\varepsilon$-envy-free allocation for this instance, every agent obtains a nonzero value.
\end{claim}

\begin{proof}
Assume on the contrary that some agent $X_0$ obtains value $0$. Note that $X_0$ (like all agents) has a block of value at least $0.2$ somewhere on the cake such that no other agent has any value there. Since the allocation is $\varepsilon$-envy-free, it follows that this block must be cut into slices of value at most $\varepsilon$. Let $X_1$ be an agent that is assigned one of the slices strictly contained in this block. Agent $X_1$ also obtains value $0$, and it must also have a block of value at least $0.2$ somewhere on the cake such that no other agent has any value there. This block must also be cut in slices of value at most $\varepsilon$, and since there are at least two slices that lie strictly inside the block, there exists such a slice that is not assigned to agent $X_0$, but rather to some agent $X_2$. We continue this procedure, always ensuring that we pick some agent that is not $X_0$ (which is always possible). Since the number of agents is finite, there exist $i < j$ such that $X_i = X_j$. If $i > 0$, then one can check that we necessarily have $X_{i-1} = X_{j-1}$. Thus, there exists $\ell > 0$ such that $X_0 = X_\ell$. However, this is impossible due to our choice of $X_\ell$, a contradiction.
\end{proof}

\paragraph*{Fixing the ordering of agents.}

\begin{claim*}
If there exists an $\varepsilon$-envy-free allocation in which agent $S_0$ gets the leftmost piece, then the \textsc{3-sat} formula is satisfiable.
\end{claim*}

\begin{proof}
In any $\varepsilon$-envy-free allocation in which agent $S_0$ gets the leftmost piece, the piece allocated to $S_0$ will be a strict prefix of $I_0[1] \cup I_0[2] = [0,2]$. Indeed, if $S_0$ were allocated all of $[0,2]$, then agent $S_0'$ would have envy $1$. It follows that agent $S_0$ will obtain value at most $1/7$. As a result, the three blocks of value $2/7$ of $S_0$ must each contain at least one cut. Also, note that the Initiation Interval $I_0$ contains at least two cuts.

We now know that the two blocks of value of $S_0$ in $I_1[1]$ and $I_1[3]$ must each contain a cut. We show that agent $S_1$ must be allocated some interval in $I_1$. Suppose for the sake of contradiction that this is not the case. Then, some agent $X_0$ must be allocated an interval in $I_1$, since there are at least two cuts inside $I_1$. But this agent cannot be $S_0$ or $S_1$, so it will obtain value $0$. However, by Claim~\ref{clm:value-larger-zero}, this is impossible.

Thus, $S_1$ must be allocated some interval in $I_1$. It follows that $S_1$ obtains value at most $0.2$. This, in turn, implies that the two blocks of value $0.4$ of $S_1$ in $I_2$ must each contain a cut. This means that we can repeat the argument above to show that $S_2$ must be allocated an interval in $I_2$. By induction it follows that every Isolating Interval contains at least $2$ cuts. Thus, we have shown that at least $2 + 2(m+n-1)=2m+2n$ cuts do not lie inside any Clause- or Variable-Gadget. This means that at most $(4m+3n)-(2m+2n)=2m+n$ cuts lie strictly inside a Clause- or Variable-Gadget, and so the Isolation property holds. By Lemma~\ref{lem:isolation}, any $\varepsilon$-envy-free allocation in which $S_0$ gets the leftmost piece will yield a satisfying assignment to the \textsc{3-sat} instance.
\end{proof}

We define the \emph{standard ordering} of allocation as follows. Starting from the left, the first piece goes to agent $S_0$ and the second piece to $S_0'$. The rest of the agents are ordered according to the order of appearance of their gadget in the instance. For this purpose, we treat every Isolating interval $I_k$ as a gadget with corresponding agent $S_k$. Within the Clause-Gadget for $C_i$, the corresponding agents appear in the order $C_i^{1}$, $C_i^{2}$, $C_i^{3}$. Within the Variable-Gadget for $x_j$, the corresponding agents appear in the order $L_j$, $R_j$. This yields a unique full ordering of all the agents in the instance.

\begin{claim*}
If the \textsc{3-sat} formula is satisfiable, then there exists an envy-free allocation in which the pieces are allocated to the agents according to the standard ordering.
\end{claim*}

\begin{proof}
Given a satisfying assignment, we show how to construct an envy-free allocation such that the pieces are allocated to the agents according to the standard ordering. Place a cut at position $1$ and through the middle of every block of $S_0$ of value $2/7$. Also place a cut through the middle of every block of value $0.4$ of $S_k$, $1 \leq k \leq m+n-2$. Allocate the leftmost piece to $S_0$ and the next piece to $S_0'$. Allocate the piece between the two cuts in the Isolating interval $I_k$ to agent $S_k$. Note that no matter how we allocate the remaining parts of the cake, the agents $S_0'$, $S_0$, $S_1$, $\dots$, $S_{m+n-1}$ will definitely be envy-free. $S_0'$ and $S_{m+n-1}$ have obtained all of their value. $S_0$ has obtained value $1/7$, but its three blocks of value $2/7$ have all been cut in half. Finally, for $1 \leq k \leq m+n-2$, $S_k$ has obtained value $0.2$, but its two blocks of value $0.4$ have also been cut in half. Figure~\ref{fig:cake:isolating-interval} shows the positions of the cuts in $I_0$, $I_1$ and $I_2$.

Depending on whether $x_j=1$ or $\overline{x}_j=1$ place a cut in the middle of the region corresponding to $x_j$ or $\overline{x}_j$ respectively inside the Variable-Gadget of $x_j$. Allocate the left piece to $L_j$ and the right piece to $R_j$. Note that $L_j$ and $R_j$ obtain all of their value.

Finally, for every clause $C_i$ pick one of its literals that is $1$ and let $C_i^k$ be the associated agent. We position two cuts inside the gadget such that $C_i^k$ gets one block of its own value, and the other two Clause-Gadget agents each get two blocks of their own value. While doing so, we also ensure that these other two agents each get one of the two remaining blocks of $C_i^k$ inside the gadget. Note that this is always possible and in fact we can also ensure that the three pieces are allocated to the agents $C_i^{1}, C_i^{2}, C_i^{3}$ in that order from left to right. $C_i^k$ has thus obtained value $0.24$ and its two other $0.24$-blocks have been allocated to two distinct agents. The last remaining block, which has value $0.28$ and lies in the corresponding Variable-Gadget, has been cut in half according to the procedure above describing how to place the cut in a Variable-Gadget. Thus, $C_i^k$ is envy-free. Now consider any of the two other agents of this Clause-Gadget. Such an agent has obtained $0.48$ of its value. $0.24$ of its value has been allocated to other agents in this Clause-Gadget, and $0.28$ of its value has been allocated to Variable-Gadget agents. Thus, this agent is also envy-free.
\end{proof}

Using these two claims we immediately obtain that the decision problems 1, 2, and 3 are NP-hard (with envy-freeness or $\varepsilon$-envy-freeness).

\paragraph*{Fixing cuts.}

\begin{claim*}
In any $\varepsilon$-envy-free allocation in which there is a cut at position $1$, the leftmost piece must be assigned to agent $S_0$.
\end{claim*}

\begin{proof}
Since there is a cut at position $1$, the leftmost piece can only contain value for agent $S_0$. Thus, by Claim~\ref{clm:value-larger-zero} it cannot be allocated to any other agent.
\end{proof}

On the other hand, given a satisfying assignment for the \textsc{3-sat} formula, we can always ensure that the corresponding envy-free allocation that we construct has a cut at position $1$. In fact, there are many more cuts that are fixed (and do not depend on what the satisfying assignment is), namely, the two cuts in each Isolating interval.

Using this observation along with the claim above, we get that the decision problems 4, 5, and 6 are NP-hard (with envy-freeness or $\varepsilon$-envy-freeness).

\section{Hardness for Indivisible Items}
\label{sec:hardness-indiv}

We now turn to a discrete analog of cake cutting, where we wish to allocate a set of indivisible items that lie on a line subject to the requirement that each agent must receive a contiguous block.
As in cake cutting, we assume that the valuations of the agents over the items are additive, and that all items must be allocated.
Besides envy-freeness, we consider the classical fairness notions of proportionality and equitability.
An allocation is \emph{proportional} if every agent receives value at least $1/n$ times her value for the whole set of items, and \emph{equitable} if all agents receive the same value.

Unlike in cake cutting, for indivisible items there may be no allocation satisfying any of the three fairness properties, e.g., when two agents try to divide a single item.
\cite{BouveretCeEl17} showed that deciding whether an envy-free allocation exists is NP-hard for additive valuations, and the same is true for proportionality; they did not consider equitability.
In this section, we extend and strengthen their results in several ways.
We consider \emph{binary valuations}, which are additive valuations such that the value of each agent for each item is either $0$ or $1$.
In other words, an agent either ``wants'' an item or not.
Even though binary valuations are much more restricted than additive valuations, as we will see, several problems still remain hard even for this smaller class.

First, we show that deciding whether a fair allocation exists is NP-hard for each of the three fairness notions mentioned.
This hardness result holds for \emph{any} non-empty combination of the three notions and even if all agents want the same number of items.
Moreover, we present a reduction that establishes the hardness for all combinations in one fell swoop.
We remark that the techniques of \cite{BouveretCeEl17} do not extend to the binary domain because each agent can have different values for different items in their construction. One may try to fix this by breaking items into smaller items to obtain a binary valuation, but each agent will require a different way of breaking items, and moreover there will be allocations in the new instance that cannot be mapped back to those in the original instance.

\begin{figure*}[!ht]
\centering
\begin{tikzpicture}[scale=1]
\draw (0,0) rectangle ++(0.8,1.2);
\draw (1,0) rectangle ++(0.8,1.2);
\draw (2,0) rectangle ++(0.8,1.2);
\draw (3,0) rectangle ++(0.8,1.2);
\draw (4,0) rectangle ++(0.8,1.2);
\node at (0.4,0.85) {$X_j$};
\node at (0.4,0.3) {$\overline{X}_j$};
\node at (1.4,0.85) {$X_j$};
\node at (1.4,0.3) {$\overline{X}_j$};
\node at (2.4,0.57) {$X_j$};
\node at (3.4,0.6) {$C_{i_p}^{k_p}$};
\node at (4.4,0.6) {$C_{i_p}^{k_p}$};
\node at (5.15,0.6) {$\dots$};
\draw (5.5,0) rectangle ++(0.8,1.2);
\draw (6.5,0) rectangle ++(0.8,1.2);
\draw (7.5,0) rectangle ++(0.8,1.2);
\draw (8.5,0) rectangle ++(0.8,1.2);
\draw (9.5,0) rectangle ++(0.8,1.2);
\node at (5.9,0.6) {$C_{i_q}^{k_q}$};
\node at (6.9,0.6) {$C_{i_q}^{k_q}$};
\node at (7.9,0.85) {$X_j$};
\node at (7.9,0.3) {$\overline{X}_j$};
\node at (8.9,0.6) {$C_{i_r}^{k_r}$};
\node at (9.9,0.6) {$C_{i_r}^{k_r}$};
\node at (10.65,0.6) {$\dots$};
\draw (11,0) rectangle ++(0.8,1.2);
\draw (12,0) rectangle ++(0.8,1.2);
\draw (13,0) rectangle ++(0.8,1.2);
\node at (11.4,0.6) {$C_{i_s}^{k_s}$};
\node at (12.4,0.6) {$C_{i_s}^{k_s}$};
\node at (13.4,0.6) {$\overline{X}_j$};
\end{tikzpicture}
\caption{Variable-Gadget for variable $x_j$. Every item is represented by a rectangle containing the agent(s) who value it.}
\label{fig:variable-gadget}
\end{figure*}
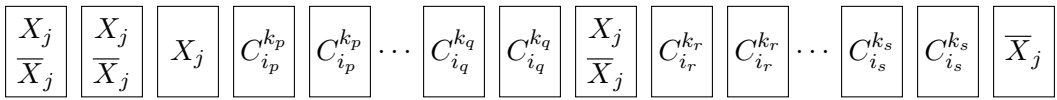

\begin{theorem}
\label{thm:indivisible-NP-combination}
Let $$F = \{\text{envy-freeness, proportionality, equitability}\},$$ and let $\emptyset\neq X\subseteq F$.
Deciding whether an instance with indivisible items on a line admits a contiguous allocation satisfying all properties in $X$ is NP-hard, even if all agents have binary valuations and value the same number of items.
\end{theorem}

\begin{proof}
We prove this result with a single reduction. 
Let $I$ be an instance of \textsc{3-SAT} with $m$ clauses $C_1, \dots C_m$ using the variables $x_1, \dots, x_n$ and their negations. 
We create the following gadgets.
\begin{itemize}
    \item Clause-Gadget: For every clause $C_i$ we introduce three agents: $C_i^1,C_i^2,C_i^3$. 
    Each of these agents is associated with one of the three literals that appear in the clause $C_i$. 
    We denote by $\ell(C_i^k)$ the literal associated with $C_i^k$. For every clause $C_i$ we construct a Clause-Gadget. 
    The gadget consists of four contiguous items that are all valued by all three agents $C_i^1,C_i^2,C_i^3$, and by no one else.
    \item Variable-Gadget: For every variable $x_j$ we introduce two agents, $X_j$ and $\overline{X}_j$, and construct a Variable-Gadget as follows (Figure~\ref{fig:variable-gadget}). 
    Starting from the left, create two items that are valued by both $X_j$ and $\overline{X}_j$ (and no one else). 
    Then, create one item that is valued only by $X_j$. Then, for every $C_i^k$ such that $\ell(C_i^k) = x_j$, create two items that are valued only by $C_i^k$.
    Then, create an item that is valued by both $X_j$ and $\overline{X}_j$. Then, for every $C_i^k$ such that $\ell(C_i^k) = \overline{x}_j$, create two items that are valued only by $C_i^k$. 
    Finally, create an item that is valued only by $\overline{X}_j$.
\end{itemize}

We combine these gadgets to create the instance $R$ as follows. Starting from the left, construct the Clause-Gadget for each clause $C_i$. Then, construct the Variable-Gadget for each variable $x_j$. Thus, we obtain an instance with $3m+2n$ agents and $4m+(5n+6m) = 5n+10m$ items.

\begin{claim*}
The following statements hold:
\begin{itemize}
    \item Any contiguous allocation in $R$ where every agent gets at least two items they value yields a satisfying assignment for $I$. This holds even if the allocation is \emph{partial}, i.e., some items are not allocated.
    \item Any satisfying assignment for $I$ yields a contiguous envy-free allocation in $R$ where every agent gets exactly two items they value.
\end{itemize}
\end{claim*}

\begin{proof}[Proof of Claim]
Consider any (possibly partial) contiguous allocation in $R$ where every agents gets at least two valued items. All of the items valued by $X_j$ or $\overline{X}_j$ lie in the Variable-Gadget for $x_j$. Let $T$ denote the second item in this gadget. Note that this item must necessarily be allocated to $X_j$ or $\overline{X}_j$ (and it cannot remain unallocated, even in a partial allocation). If $X_j$ obtains $T$, then we set $a_j=1$. If $\overline{X}_j$ obtains $T$, we set $a_j=0$. We now claim that $a$ is a satisfying assignment for $I$. Consider any clause $C_i$ and the three associated agents $C_i^1,C_i^2,C_i^3$. At most two of those agents can obtain their two items from the Clause-Gadget for $C_i$. Thus, there exists $k \in [3]$ such that $C_i^k$ is allocated a valued item outside the Clause-Gadget. But the only other place where $C_i^k$ values items is inside the Variable-Gadget for the variable of $\ell(C_i^k)$ (the literal in clause $C_i$ corresponding to agent $C_i^k$). Since $C_i^k$ obtains an item in this gadget, one can check that the agent corresponding to the literal $\ell(C_i^k)$ must obtain the second item in the gadget. It follows that the literal $\ell(C_i^k)$ has value $1$ in the assignment $a$, and thus the clause $C_i$ is satisfied by $a$.

Conversely, let $a$ be any satisfying assignment for $I$. For every clause $C_i$, there exists an agent $C_i^k$ such that the literal $\ell(C_i^k)$ is true in $a$. Allocate the four items in the Clause-Gadget for $C_i$ to the other two clause agents (two contiguous items for each). Then, $C_i^k$ has only two valued items remaining, namely the ones in the Variable-Gadget corresponding to $\ell(C_i^k)$. Allocate them to $C_i^k$. Once this is done for all clauses, we move on to the Variable-Gadget agents. Assume that $a_j=1$; the case where $a_j=0$ can be treated analogously. Then, the first two items of the Variable-Gadget for $x_j$ are allocated to $X_j$, while $\overline{X}_j$ obtains the only two remaining items that it values (which are not adjacent). However, no clause agent $C_i^k$ has been allocated any item in this interval, because items there are only valued by $C_i^k$ with $\ell(C_i^k) = \overline{x}_j$ and those agents have been allocated items within their respective Clause-Gadget (because $a_j=1$); therefore we may allocate all items in this interval to $\overline{X}_j$. At this point, some items in the Variable-Gadget might still be unallocated, namely items that lie in the interval starting from the third item up to the last item not allocated to $\overline{X}_j$. If all of these items are unallocated, then allocate them all to $\overline{X}_j$. Note that the items allocated to $\overline{X}_j$ are indeed contiguous. If some of these items are already allocated, then they are allocated to clause agents. Simply extend the intervals allocated to these clause agents until they form a partition of this region. This construction ensures that every agent $A$ obtains exactly two items they value.
Moreover, for every other agent $B$, $A$ obtains at most two items valued by $B$.
\end{proof}

The final step of the proof is to introduce one last gadget. The Special-Gadget creates $3m+2n+7$ new agents. We denote the set of these new agents by $N$. The gadget consists of $2(3m+2n)+14=6m+4n+14$ new items. These items are valued by all agents in $N$. For every $i \in [m]$ and $k \in [3]$, $C_i^k$ values all new items except the rightmost six. For every $j \in [n]$, $X_j$ and $\overline{X}_j$ value all new items except the rightmost four.

The Special-Gadget is added to the right end of $R$ and yields the final instance $R'$. Note that in $R'$ there are $6m+4n+7$ agents and every agent values exactly $6m+4n+14$ items. Now consider any contiguous allocation for $R'$.

\begin{itemize}
    \item If the allocation is proportional, then every agent gets at least $\lceil (6m+4n+14)/ (6m+4n+7)\rceil = 2$ items they value. It follows that the agents in $N$ get all the new items, because $2|N|=2(3m+2n+7)=6m+4n+14$. This means that the other agents get at least two items they value in $R$. By the claim above, we obtain a satisfying assignment.
    \item If the allocation is equitable, then all agents get exactly $s$ items they value, for some $s \geq 0$. The Special-Gadget contains an item (in fact, many) that is valued by all agents. Since this item will be allocated to someone, $s=0$ is not possible. Also $s \geq 3$ is not possible, because the $3m+2n+7$ agents in $N$ all like the exact same $2(3m+2n+7)$ items. Now, since all $6m+4n+7$ agents value the first $(6m+4n+14) - 6 = 6m+4n+8$ items in the Special-Gadget, at least one of them will be allocated to two of those (by the pigeonhole principle). It follows that $s=1$ is also impossible. Thus, only $s=2$ remains, and we again obtain a satisfying assignment by the claim.
\end{itemize}

Since envy-freeness implies proportionality, it follows that any $X$-allocation for $R'$ yields a satisfying assignment for the \textsc{3-SAT} instance $I$, for any non-empty $X \subseteq \{$envy-free, proportional, equitable$\}$. On the other hand, any satisfying assignment for the \textsc{3-SAT} instance yields an envy-free and equitable allocation for $R'$, by assigning two contiguous Special-Gadget items to each agent in $N$ and then using the claim.
\end{proof}

In the construction used for our proof of Theorem~\ref{thm:indivisible-NP-combination}, each agent values at most four contiguous block of items.
In light of this result, one may naturally wonder whether the hardness continues to hold if, for example, every agent values a single block of items.
We show that this is the case for proportionality, provided that we drop the requirement that all agents value the same number of items.
Note that if each agent values a contiguous block of items \emph{and} all agents value the same number of items, deciding whether a proportional allocation exists can in fact be done in polynomial time.
Indeed, we can view the problem as a scheduling problem on a single machine, with each agent having a task to be completed by a machine.
For a given task, its \emph{release time} is where the corresponding agent's valued block starts, its \emph{deadline} is where the block ends, and its \emph{length} is the number of items that we need to give the agent in order to satisfy proportionality.
When all tasks have the same length, which is true in our setting, polynomial-time algorithms have been proposed by \cite{Simons78} and \cite{GareyJoSi81}.

\begin{theorem}
\label{thm:indivisible-proportional}
Deciding whether an instance with indivisible items on a line admits a contiguous proportional allocation is NP-hard, even if the valuations are binary and every agent values a contiguous block of items.
\end{theorem}

\begin{proof}
We reduce from the {\normalfont \scshape 3-partition} problem. 
An instance of the {\normalfont \scshape 3-partition} problem consists of $3n$ positive integers $x_1,\dots,x_{3n}$ with sum $nB$, and the goal is to partition them into $n$ sets of size three each so that the three numbers in each set sum to $B$.
The problem is NP-hard, and remains so when $B/4<x_i<B/2$ for all $i$ \citep{GareyJo79}.

Given an instance of {\normalfont \scshape 3-partition}, we create an instance of our problem as follows. 
There are $m := n(B+1)+4nk^2$ items on the line, where $k=4B$.
Each item belongs to one of the three types: special, normal, and dummy.
From left to right, the last $4nk^2$ items are dummy items.
The remaining $n(B+1)$ items are partitioned into $n$ blocks of size $B+1$---the leftmost item of each block is a special item (so $n$ special items in total), and the remaining $B$ items of the block are normal items (so $nB$ normal items in total).
There are $n' := 4n(k+1)$ agents: $n$ special, $3n$ normal, and $4nk$ dummy.
Each of the $n$ special agents values a distinct special item and nothing else.
Each dummy agent values all dummy items and nothing else.
For $1\leq i\leq 3n$, the $i$th normal agent values the leftmost $n'x_i$ items.
Note that this is well-defined because $n'x_i < 2n(k+1)B < 4nkB = nk^2 < m$.
Moreover, $n'x_i > n(k+1)B  > 2nB > n(B+1)$, so each normal agent values all normal items (along with other items).

First, suppose that there is a valid solution to the {\normalfont \scshape 3-partition} instance. 
We construct a proportional allocation.
Give each special agent her valued item, and each dummy agent $k$ consecutive dummy items.
For each part $\{x_{a_1},x_{a_2},x_{a_3}\}$ in the solution to the {\normalfont \scshape 3-partition} instance, we pick a block of $B$ normal items and give $x_{a_i}$ consecutive items to the $a_i$th normal agent.
One can check that the resulting allocation is proportional; in particular, each dummy agent needs at least $\left\lceil\frac{4nk^2}{n'}\right\rceil = \left\lceil\frac{4nk^2}{4n(k+1)}\right\rceil = k$ valued items, and that is exactly what they get.

Conversely, suppose that our construction admits a proportional allocation.
In this allocation, each special agent must get her valued item and, as above, each dummy agent needs at least $k$ valued items. Since there are $4nk$ dummy agents and they value the same $4nk^2$ items, each dummy agent must receive exactly $k$ valued items.
This leaves only the $nB$ normal items to be allocated to the $3n$ normal agents.
Normal agent $i$ needs to get at least $x_i$ items, so given that $\sum_{i=1}^{3n}x_i = nB$, all normal items must be allocated to the normal agents, and normal agent $i$ must receive exactly $x_i$ items.
Finally, since $B/4 < x_i < B/2$ for all $i$, each block of $B$ normal items is allocated to exactly three agents.
Hence the allocation yields a valid solution to the {\normalfont \scshape 3-partition} instance, as desired.
\end{proof}

Next, we show that under the same conditions as Theorem~\ref{thm:indivisible-proportional}, deciding whether there exists a proportional and equitable allocation, or an equitable allocation that gives the agents positive value, are both computationally hard.
Since agents do not all value the same number of items (unlike in Theorem~\ref{thm:indivisible-NP-combination}), we normalize the valuations so that if agent $i$ values $x_i$ items, she has value $1/x_i$ of each of them (so her total value is $1$).

\begin{theorem}
\label{thm:indivisible-equitable}
Deciding whether an instance with indivisible items on a line admits
\begin{itemize}
\item a contiguous allocation that is both proportional and equitable;
\item a contiguous equitable allocation in which the agents receive positive value
\end{itemize}
are both NP-hard, even if the valuations are binary and every agent values a contiguous block of items.
\end{theorem}

\begin{proof}
The reduction is similar to the one in Theorem~\ref{thm:indivisible-proportional}. We again reduce from \textsc{3-partition}, but this time we also assume that $x_i > n$ for all $i$. Note that we can ensure that this is the case by multiplying all $x_i$ and $B$ by $n$.

Let $K=nB$. The main building block of this reduction is a \emph{$K$-block}: $K$ consecutive items with $K$ agents who only value these $K$ items. The instance is constructed as follows. Starting from the left end of the line, there are $B$ consecutive $K$-blocks. Note that each $K$-block has its own $K$ agents. We call this the ``left region'' of the instance. The ``right region'' of the instance consists of $n$ blocks of $K+B$ items each. The leftmost $K$ items of such a block form a $K$-block, and there are $B$ items to the right of that $K$-block. Finally, we introduce new agents $a_1, \dots, a_{3n}$. For each $i \in [3n]$, agent $a_i$ values the $Kx_i$ rightmost items on the line. Note that this is well-defined, since there are $BK + n(K+B) \geq BK\geq Kx_i$ items overall. Furthermore, agent $a_i$ values all items in the right region, because $Kx_i \geq n(K+B)$ (since $K=nB$ and $x_i > n$). Note that every agent values a contiguous block of items.

Now consider any equitable allocation in which the agents receive positive value. Every agent must get at least one item that they value. Consider any $K$-block. Since its $K$ agents only value these $K$ items, it follows that they each obtain exactly one. Thus, they each get value exactly $1/K$, and all other agents in the instance must also get value exactly $1/K$. This means that agent $a_i$ must obtain exactly $x_i$ of its valued items. 
Since $B/4<x_i<B/2$ for all $i$, each block of $B$ items in the right region are allocated to exactly three agents $a_i$.
Hence, we obtain a solution to the \textsc{3-partition} instance.
Note that a proportional and equitable allocation yields positive value to the agents, so it also gives rise to a solution to the \textsc{3-partition} instance.

Conversely, given a solution to the \textsc{3-partition} instance, one can construct an equitable allocation in which the agents receive positive value by following the previous paragraph. Note that this allocation is also proportional, since each agent receives value $1/K$ and there are more than $K$ agents. This completes the proof.
\end{proof}

Finally, we consider \emph{approximate} envy-freeness for the discrete setting as well. We show that for a sufficiently small constant $\varepsilon$, deciding whether there exists an $\varepsilon$-envy-free allocation is NP-hard; this holds even if we restrict the valuation functions as in Theorem~\ref{thm:indivisible-NP-combination}.

\begin{theorem}
\label{thm:indivisible-NP-envyfree}
For any $\varepsilon < 1/13$, deciding whether a contiguous $\varepsilon$-envy-free allocation exists is NP-hard, even if all agents have binary valuations and value the same number of items.
\end{theorem}

\begin{proof}
Consider an instance of \textsc{3-SAT} with $m$ clauses $C_1, \dots, C_m$ using the variables $x_1, \dots, x_n$ and their negations. We will make use of the following gadgets:
\begin{itemize}
    \item Clause-Gadget: For every clause $C_i$ we introduce three agents $C_i^1, C_i^2, C_i^3$. Each of these agents is associated with one of the three literals appearing in clause $C_i$, and we denote by $\ell(C_i^k)$ the literal associated to $C_i^k$. For every clause $C_i$ we construct a Clause-Gadget as follows. Starting from the left, the first three items are valued by $C_i^1$ (and no one else), the next three items are valued by $C_i^2$, and the next three by $C_i^3$. We repeat this three times. Thus, the Clause-Gadget for $C_i$ consists of $27$ items and every agent $C_i^k$ values exactly $9$ of these items.
    \item Variable-Gadget: For every variable $x_j$ we introduce two agents $L_j$ and $R_j$ and construct the Variable-Gadget for $x_j$ as follows. Starting from the left, the first $13$ items are valued by $L_j$. The next $4$ items are valued by every agent $C_i^k$ such that $\ell(C_i^k) = x_j$ (i.e., every agent corresponding to the literal $x_j$). The next $4$ items after that are valued by every agent $C_i^k$ such that $\ell(C_i^k) = \overline{x}_j$. Finally, the next $13$ items are valued by $R_j$.
    \item Isolation-Gadget: An Isolation-Gadget consists of $13$ items and $5$ agents. The $5$ agents value each of the $13$ items and no other items in the instance.
\end{itemize}
The instance is then constructed as follows. Starting from the left, we construct the Clause-Gadget for $C_1$, then for $C_2$, and so on up to $C_m$. Then, we construct the Variable-Gadget for $x_1$, for $x_2$, and so on up to $x_n$. Finally, we introduce an Isolation-Gadget between any two adjacent gadgets. Thus, there are $m+n-1$ Isolation-Gadgets, and the instance has $3m+2n+5(m+n-1) = 8m+7n-5$ agents.

Note that in this construction every agent values exactly $13$ items. Since all of the valuations are binary, this means that for the normalized valuations, any $\varepsilon$-envy-free allocation with $\varepsilon < 1/13$ must actually be (exactly) envy-free.

Consider any contiguous envy-free allocation for this instance. The Variable-Gadget for $x_j$ must contain at least one cut strictly inside it---otherwise, $L_j$ or $R_j$ would not be envy-free. Furthermore, the Clause-Gadget for $C_i$ must strictly contain at least two cuts. If it contained at most one cut, then at least one of the agents $C_i^k$ would not obtain any item in this gadget. Thus, $C_i^k$ would be able to obtain at most $4$ valued items (from the Variable-Gadget for the variable of $\ell(C_i^k)$). However, since the Clause-Gadget for $C_i$ has been divided into at most two pieces and $C_i^k$ values $9$ items in this gadget, one of those pieces contains at least $5$ items valued by $C_i^k$. Thus, $C_i^k$ would envy the agent receiving that piece.

Finally, any Isolation-Gadget must strictly contain at least $6$ cuts. It is easy to see that it must contain at least $4$ cuts, so that each of the $5$ agents that values all of the $13$ items can obtain something. However, $4$ cuts are not enough, because the $5$ resulting pieces cannot contain exactly the same number of items and thus one of the $5$ agents would not be envy-free. It turns out that $5$ cuts are also not enough. Indeed, in that case there are $6$ pieces and $5$ of those must be given to the $5$ agents of the gadget. However, it is impossible to divide $13$ items into $6$ pieces in such a way that $5$ of the pieces contain the same number of items and the $6$th piece contains at most that many items.

Since the instance has $8m+7n-5$ agents, there are $8m+7n-6$ cuts. With the arguments above we have accounted for exactly $2m+n+6(m+n-1) = 8m+7n-6$ cuts. It follows that every Clause-Gadget strictly contains exactly $2$ cuts and every Variable-Gadget strictly contains exactly $1$ cut. Thus, similarly to the divisible case, we have ensured that a certain Isolation property holds. The proof that this allocation yields a satisfying assignment for the \textsc{3-SAT} instance is analogous to the divisible case (Lemma~\ref{lem:isolation}).

Conversely, given a satisfying assignment of the \textsc{3-SAT} instance, it is not hard to construct an envy-free contiguous allocation for the instance. In fact, the only difference from the divisible case is with respect to the Isolation-Gadgets. Here, the $6$ cuts inside every Isolation-Gadget are placed as follows: place a cut after the first item, then a cut every two items, and give the $5$ central pieces of size $2$ to the $5$ agents of the gadget.
\end{proof}

\section{Connections Between Various Cake-Cutting Problems}\label{sec:connections}

In this section, we uncover several new connections between different cake-cutting settings. In particular, in Section~\ref{sec:exact-to-approx} we show that for piecewise constant valuations, finding an approximate envy-free allocation is as hard as finding an exact one. Then, in Section~\ref{sec:continuous-and-discrete} we exhibit connections between a number of continuous and discrete cake-cutting problems.

\subsection{Approximate and Exact Envy-Freeness}\label{sec:exact-to-approx}

We begin by proving the following result, which relates approximate and exact envy-freeness for a restricted yet quite expressive class of valuations.

\begin{theorem}\label{thm:exactEFtoeps}
For piecewise constant valuations, computing a contiguous envy-free allocation reduces to computing a contiguous $\varepsilon$-envy-free allocation for a sufficiently small $\varepsilon$ (which may depend on the number of agents and the valuations).
\end{theorem}

In particular, this means that for such valuations there always exists a contiguous envy-free allocation in which all cut points are rational.\footnote{This is not the case for more general valuations~\citep{Stromquist08}.} Theorem~\ref{thm:exactEFtoeps} is implied by the following result:
\begin{lemma}
Let $v_1, \dots, v_n$ be (explicit, normalized) piecewise constant valuations, and $M\geq 3$ and $k$ be positive integers such that
\begin{itemize}
    \item for all $i \in [n]$, all of the numbers in the explicit description of $v_i$ (i.e., the step heights and step change positions) have numerator and denominator at most $M$;
    \item for all $i \in [n]$, $v_i$ has at most $k$ value-blocks.
\end{itemize}
Then from any $M^{-20kn}$-envy-free solution, we can efficiently obtain an envy-free solution.
\end{lemma}

\begin{proof}
We apply the technique that was used by~\cite{EtessamiYa10} to show that finding an exact fixed point of a LinearFIXP circuit reduces to finding a (sufficiently good) approximate fixed point.

Let $(\hat{x},\pi)$ be a contiguous $\varepsilon$-envy-free allocation, where $0 \leq \hat{x}_1 \leq \dots \leq \hat{x}_{n-1} \leq 1$. Without loss of generality assume that $\pi(i)=i$ for all $i$ (by reordering the agents). For $i \in [n]$ and $j \in [n-1]$, let $\ell_j^i$ be the position of the closest step change in $v_i$ on the left of $\hat{x}_j$. Similarly, define $r_j^i$ to be the closest step change position on the right. (If $\hat{x}_j$ lies on a step change position of $v_i$, then we set $\ell_j^i = r_j^i = \hat{x}_j$.) Finally, set $\ell_j = \max_i \ell_j^i$ and $r_j = \min_i r_j^i$. Note that all valuation densities are constant on the interval $[\ell_j,r_j]$. Thus, the corresponding cumulative valuation functions are linear. For $i \in [n]$ and $j \in [n-1]$, let $h_j^i$ denote the (constant) value of the density function of $v_i$ in $[\ell_j,r_j]$ (if $\ell_j=r_j$ then $h_j^i$ can be defined as any arbitrary value).

We solve the following linear program (LP) with variables $x_1, \dots, x_{n-1},z$:
\bigbreak
\begin{tabular}{rl}
    $\min z$ & \\
    $\ell_j \leq x_j \leq r_j$ & $\forall j \in [n-1]$ \\
    $x_j \leq x_{j+1}$ & $\forall j \in [n-2]$ \\
    $\left[h_j^i(x_j-\ell_j) - h_{j-1}^i(x_{j-1}-\ell_{j-1}) + v_i(\ell_{j-1},\ell_j)\right]$ \phantom{- $\leq z$} & \\
    $- \left[h_i^i(x_i-\ell_i) - h_{i-1}^i(x_{i-1}-\ell_{i-1}) + v_i(\ell_{i-1},\ell_i)\right] \leq z$ & $\forall i,j \in [n]$
\end{tabular}
\bigbreak
\noindent where we define $x_0=\ell_0=0$ and $x_n=\ell_n=1$ for ease of exposition. Note that the left-hand side of the last line is equal to $v_i(x_{j-1},x_j) - v_i(x_{i-1},x_i)$, i.e., the envy of agent $i$ towards agent $j$. Thus, minimizing $z$ corresponds to minimizing the maximum envy experienced by any agent. The LP does the following: it allows any cut in $(\hat{x},\pi)$ to move slightly to the left or the right, as long as it does not fall into a different step (in any of the valuations) and as long as the relative order of the cuts does not change. Furthermore, the order of assignment of the intervals to the agents (i.e., $\pi$) does not change.

Clearly, $(\hat{x},\varepsilon)$ is a feasible solution of the LP. For now assume that we know that the LP has an optimal (rational) solution $(x^*,z^*)$ such that all denominators are bounded by some positive integer $d$ (that only depends on $M$, $k$ and $n$). Then, if we pick $\varepsilon < 1/d$, it will follow that $z^* < 1/d$, which implies that $z^*=0$ ($z^* \geq 0$ is implicitly forced by the constraints). Thus, solving the LP will give us a contiguous envy-free allocation.

It remains to find a bound $d$ such that the LP is guaranteed to have an optimal solution with all denominators bounded by $d$. The LP must have a solution $(x^*,z^*)$ that is a vertex of the feasible polytope---the polytope defined by the constraints. Note that for $(x^*,z^*)$, at least $n$ constraints must be tight, i.e., satisfied with equality. Furthermore, $(x^*,z^*)$ must be the unique point that satisfies all these tight constraints with equality (otherwise, it would not be a vertex of the feasible polytope). Thus, by picking a linearly independent subset of these tight constraints, we get that $y=(x^*,z^*)$ is the unique solution of a linear system $Ay=b$ with $n$ variables and $n$ equations.

In order to investigate the denominator size of solutions of $Ay=b$, we first turn it into an equivalent linear system $A'y=b'$ where $A'$ and $b'$ are integral. Specifically, we will multiply the $m$th line of the linear system by some value $C_m$ so that all of the coefficients become integers. Inspection of the LP reveals that any line $m$ contains at most $6$ non-zero entries, i.e., at most $6$ entries out of $a_{m,1}, \dots, a_{m,n}, b_m$ are not zero. Furthermore, at most $5$ of them are non-integral, because the coefficient of $z$ is integral. Out of these, at most one (namely $b_m$) can have a denominator that is larger than $M$. This corresponds to the case where $b_m = -h_j^i \ell_j + h_{j-1}^i \ell_{j-1} + h_i^i \ell_i - h_{i-1}^i \ell_{i-1} + v_i(\ell_{j-1},\ell_j) - v_i(\ell_{i-1},\ell_i)$.

The first 4 terms in the expression of $b_m$ have denominator at most $M^2$ (since they are of the form $p_1 \cdot p_2$, where $p_1,p_2$ have denominator at most $M$). The term $v_i(\ell_{j-1},\ell_j)$ can be computed by summing up the values of the blocks between $\ell_{j-1}$ and $\ell_j$ with respect to $v_i$. The value of each block has denominator at most $M^3$ (since it is of the form $(p_1-p_2)\cdot p_3$, where $p_1,p_2,p_3$ have denominator at most $M$). Since there are at most $k$ blocks in $v_i$, the denominator of $v_i(\ell_{j-1},\ell_j)$ is at most $M^{3k}$. The same also holds for $v_i(\ell_{i-1},\ell_i)$. Thus, the denominator of $b_m$ is at most $M^{6k+8}$.

It follows that there exists some integer $C_m \leq M^4\cdot M^{6k+8} = M^{6k+12}$ such that multiplying the $m$th line of the linear system by $C_m$ makes the coefficients integral. Doing this for every line yields an equivalent linear system $A'y=b'$ that is integral. Notice that $A'$ has at most 5 non-zero entries per line and each of these values is bounded (in absolute value) by $M^{6k+13}$. Cramer's rule tells us that $z^* = \frac{\det (C)}{\det (A')}$, where $C$ is the matrix $A'$ with the last column replaced by $b'$. Since $\det(C)$ is an integer, it suffices to bound $|\det(A')|$ in order to bound the denominator of~$z^*$.

Using Hadamard's inequality, we get that $|\det(A')| \leq \prod_{m=1}^n \|A'_m\|_2$, where $A'_m$ is the $m$th line (i.e., row) of $A'$. It follows that $\|A'_m\|_2 \leq \sqrt{5} M^{6k+13} \leq M^{6k+14}$ (since $M \geq 3$). Thus, we get that $z^*$ has denominator at most $d := M^{(6k+14)n} \leq M^{20kn}$.
\end{proof}

\begin{remark*}
The same proof also yields the following result: If for all $i \in [n]$, all numbers in the description of $v_i$ have denominator \emph{exactly} $M$, then from any $M^{-4n}$-envy-free solution, we can efficiently obtain an envy-free solution. Indeed, in this case $b_m$ has denominator $M^2$, so we get $C_m=M^2$ and $\|A'_m\|_2 \leq M^4$ for every $m$.
\end{remark*}

\subsection{Continuous and Discrete Cake Cutting}\label{sec:continuous-and-discrete}

We now establish the computational equivalence between some continuous and discrete cake-cutting problems. Let us start by defining the computational problems that we will consider.

\begin{definition}
The problem \textsc{unary-$\varepsilon$-EF-Cake-Cutting} is defined as: given $\varepsilon > 0$ (\emph{in unary}) and (explicit, normalized) piecewise constant valuations $v_1, \dots, v_n$ on $[0,1]$, find a contiguous $\varepsilon$-envy-free allocation $(x, \pi)$.
\end{definition}

This corresponds to the standard contiguous $\varepsilon$-envy-free cake-cutting problem with piecewise constant valuations, except that $\varepsilon$ is provided in unary representation. This means that $\varepsilon$ can no longer have exponential precision with respect to the size of the input. We also define a (seemingly) more restricted version of this problem.

\begin{definition}
The problem \textsc{simple-$\varepsilon$-EF-Cake-Cutting} is defined exactly as \textsc{unary-$\varepsilon$-EF-Cake-Cutting}, except that we are also given some positive integer $M$ (\emph{in unary}) and for all $i \in [n]$ we have that the piecewise constant valuation $v_i$ satisfies:
\begin{itemize}
    \item all heights of value-blocks of $v_i$ are integral;
    \item the height of $v_i$ can only change at points of the form $k/M$ where $k \in [M]$.
\end{itemize}
\end{definition}

Next, we consider discrete cake cutting. While an envy-free allocation is not guaranteed to exist in this setting (cf. Section~\ref{sec:hardness-indiv}), such an allocation always exists for some restricted classes of valuations.
We say that indivisible item valuations $v_1, \dots, v_n$ are \emph{disjoint} if every item is valued by at most one agent. \citet{MarencoTe14} proved that if the valuations are disjoint, then an envy-free allocation necessarily exists. We define a computational search problem based on this existence theorem, where we restrict ourselves to the binary valuation case. Note that binary valuations correspond to piecewise uniform valuations once normalized (i.e., if an item is valued by an agent, then it is valued the same as any other item valued by that agent).

\begin{definition}
The problem \textsc{Disjoint-Discrete-EF-Cake-Cutting} is defined as: given disjoint binary valuations $v_1, \dots, v_n$ on a discrete cake, find a contiguous envy-free allocation.
\end{definition}

Perhaps surprisingly, it turns out that all of these problems are computationally equivalent. Thus, any algorithm or hardness result for one of them would immediately extend to all of them.

\begin{theorem}
The following problems are polynomial time equivalent:
\begin{enumerate}
    \item[(1)] \textsc{unary-$\varepsilon$-EF-Cake-Cutting}
    \item[(2)] \textsc{simple-$\varepsilon$-EF-Cake-Cutting}
    \item[(3)] \textsc{Disjoint-Discrete-EF-Cake-Cutting}
    \item[(4)] \textsc{Disjoint-Discrete-EF-Cake-Cutting}, where all agents value the same number of items.
\end{enumerate}
\end{theorem}

The rest of this section is devoted to proving this theorem. The reductions (2) $\rightarrow$ (1) and (4) $\rightarrow$ (3) are trivial, because we are reducing from a special case of a problem to a more general case. Thus, in order to establish the theorem, it remains to show that (1) reduces to (4) (Proposition~\ref{prop:continuous-to-discrete}), and that (3) reduces to (2) (Proposition~\ref{prop:discrete-to-continuous}).

\begin{proposition}\label{prop:continuous-to-discrete}
\textsc{unary-$\varepsilon$-EF-Cake-Cutting} reduces to \textsc{Disjoint-Discrete-EF-Cake-Cutting} where all agents value the same number of items.
\end{proposition}

\begin{proof}
We follow the same idea that was used by \citet{Filos-RatsikasGo18} to show that  \textsc{$\varepsilon$-Consensus-Halving} reduces to \textsc{Necklace-Splitting} when $\varepsilon$ is given in unary representation (i.e., it is inversely polynomial).

Let $m$ denote the maximum number of value-blocks in the piecewise constant valuation of any agent $1 \leq i \leq n$. Since the piecewise constant valuations are provided explicitly in the input, it follows that $m$ is bounded by the size of the input. Let $\delta \leq \varepsilon/(m+2)$ be such that $1/\delta$ is integral.

For each agent $i$ and each value-block of $v_i$ we do the following. Let $[a,b]$ denote the subinterval covered by the block and let $h$ be its height. We divide the block into sub-blocks of value $\delta$ each, starting from the left. Namely, the first sub-block covers $[a,a+\delta/h]$, the second sub-block covers $[a+\delta/h,a+2\delta/h]$, and so on. If $(b-a)h/\delta$ is not an integer, then the last sub-block will be incomplete and we will ignore it. Thus, we have obtained $\lfloor (b-a)h/\delta \rfloor$ complete sub-blocks. For each such sub-block, we compute its midpoint and place an item valued by agent $i$ at that position in $[0,1]$.

After we have done this for every block of every agent, we perform some post-processing. Note that all agents might not value the same number of items. Indeed, since incomplete sub-blocks are dropped, an agent might value less than $1/\delta$ items. However, since every agent has at most $m$ blocks of value, she can have at most $m$ incomplete sub-blocks. Thus, every agent values at least $1/\delta - m$ items. Now, for any agent that values strictly more than $1/\delta - m$ items, we remove items from the instance until she values exactly $1/\delta - m$ items. The items to be removed are picked arbitrarily---since every item is valued by exactly one agent, this is straightforward to do. After this is done, every agent values exactly $1/\delta - m$ items. In particular, exactly $m \delta$ of every agent's original value is unaccounted for by the discretized instance.

From here we obtain an instance of \textsc{Disjoint-Discrete-EF-Cake-Cutting} by simply arranging the items in the order in which they appear in the interval $[0,1]$. Note that it is possible that items have the exact same position in $[0,1]$---in that case, we resolve the tie arbitrarily. Every item is valued $1/(1/\delta - m)$ by exactly one agent, and $0$ by all other agents.

Consider any solution of this \textsc{Disjoint-Discrete-EF-Cake-Cutting} instance. This allocation of the items gives rise to an allocation of the cake: for every cut in the discretized version, we place the corresponding cut in the continuous version between the positions of the items on either side of the cut (e.g., halfway between the positions of the two items). In particular, if the two items share the same position in $[0,1]$, then the cut is placed at that same position.

We now argue that the resulting allocation is an $\varepsilon$-approximate solution to the \textsc{unary-$\varepsilon$-EF-Cake-Cutting} instance. Let $v_{ij}$ denote the $i$-value (i.e., the value for agent $i$) of the interval assigned to agent $j$. Let $V_{ij}$ denote the value for agent $i$ of the items assigned to agent $j$ in the discretized instance, but where we let every item have value $\delta$ (instead of $1/(1/\delta - m)$). Then, we have $V_{ii} \geq V_{ij}$ for all $i,j$. Consider the interval assigned to agent $i$ and compare it to the items assigned to agent $i$. It is possible that even though an item was assigned to agent $i$, the cut in the continuous instance cuts through the corresponding sub-block of value $\delta$. However, in that case agent $i$ gets at least $\delta/2$ from that sub-block, i.e., she lost at most $\delta/2$. Since this can happen at both extremities of the interval assigned to agent $i$, we get $v_{ii} \geq V_{ii}-\delta$. Now consider the $i$-value of the  interval assigned to agent $j$. The same idea as above about the extremities of the interval means that the continuous allocation might increase the $i$-value by $\delta$ with respect to the discrete allocation. Furthermore, there is also $m \delta$ of agent $i$'s available value that is unaccounted for in the discrete allocation. In the worst case, all of it lies in the interval allocated to agent $j$. Thus, we obtain $v_{ij} \leq V_{ij} + \delta + m \delta$. Putting everything together, we then get $v_{ii} \geq v_{ij} -(m+2)\delta \geq v_{ij} - \varepsilon$.
\end{proof}

\begin{proposition}\label{prop:discrete-to-continuous}
\textsc{Disjoint-Discrete-EF-Cake-Cutting} reduces to \textsc{simple-$\varepsilon$-EF-Cake-Cutting}.
\end{proposition}

\begin{proof}

Consider an instance of \textsc{Disjoint-Discrete-EF-Cake-Cutting} with $m$ items and $n$ agents with disjoint binary valuations $v_1, \dots, v_n$. For $i \in [n]$, let $m_i$ denote the number of items that agent $i$ values. Note that the valuations are provided explicitly in the input, so $n$, $m$, and the $m_i$'s are bounded by the size of the input. We start by providing a reduction to \textsc{unary-$\varepsilon$-EF-Cake-Cutting}.

We construct a continuous cake-cutting instance as follows. Divide the continuous cake $[0,1]$ into $m$ regions of size $1/m$, i.e., $I_j=[(j-1)/m,j/m]$ for $j \in [m]$. If item $j \in [m]$ is valued by agent $i$, then put a block of length $1/m$ and height $m/m_i$ in interval $I_j$ of the valuation $w_i$. This yields piecewise constant valuations $w_1, \dots, w_n$ on $[0,1]$ that are normalized. Finally, set $\varepsilon := \min_i 1/(nm_i)$. Note that $\varepsilon$ can be efficiently represented in unary.

Let $(x,\pi)$ be a contiguous $\varepsilon$-envy-free allocation for this continuous cake-cutting instance. For $i \in [n]$, let $A_i$ denote the interval allocated to agent $i$. We now provide a rounding procedure to turn this $\varepsilon$-envy-free allocation into an envy-free allocation where all cuts lie on points of the form $j/m$ with $j \in [m]$. It is easy to see that this yields a solution to the \textsc{Disjoint-Discrete-EF-Cake-Cutting} instance.

Consider any ``bad cut'', i.e., a cut that lies strictly within an interval $I_j=[(j-1)/m,j/m]$. If item $j$ is not valued by any agent, then we can move this cut to either side without changing the value of any agent for any of the allocated intervals. Otherwise, item $j$ is valued by exactly one agent $i$; in that case we say that this cut is an \emph{$i$-cut}. If the interval $A_i$ allocated to agent $i$ has an $i$-cut on the left or right endpoint, then we can immediately round it in favour of agent $i$, i.e., making $A_i$ larger. Indeed, this does not decrease $w_i(A_i)$, does not increase $w_i(A_\ell)$, and does not change $w_\ell(A_k)$ for any $k,\ell \in [n]$ with $\ell \neq i$. Thus, we still have an $\varepsilon$-envy-free allocation and $w_i(A_i)$ is of the form $t_i/m_i$ where $t_i \in [m_i]$ (and will no longer change).

If we had an envy-free allocation, then the rest of the rounding would be straightforward: simply round every remaining bad cut to the right (or alternatively round every bad cut to the nearest $\cdot/m$ value). With this rounding, agent $i$ would have envy strictly less than $1/m_i$, and thus envy $0$. However, we only have an $\varepsilon$-envy-free allocation, so we need to do come up with a more involved rounding scheme. We now show how to round all $i$-cuts; the same procedure can be applied for every $i \in [n]$.

An \emph{$i$-chain} is a set of consecutive cuts that are all $i$-cuts, and such that it has maximal length (i.e., the first cut to the left and right is not an $i$-cut). Every $i$-chain can be handled separately as follows. Consider an $i$-chain consisting of $k$ consecutive $i$-cuts. Let $A_{\ell_1}, \dots, A_{\ell_{k+1}}$ be the allocated intervals delimited by those cuts, from left to right. Note that $\ell_{1}, \dots, \ell_{k+1} \neq i$.

We start with the leftmost $i$-cut: the one lying between $A_{\ell_1}$ and $A_{\ell_2}$. Suppose that it lies strictly within the interval $[(j-1)/m, j/m]$. There are two possible rounding positions for this cut: to the left or to the right. If we round to the left then $w_i(A_{\ell_1}) = t/m_i$ for some $t$, while if we round to the right we get $w_i(A_{\ell_1}) = (t+1)/m_i$. If $(t+1)/m_i \leq t_i/m_i = w_i(A_i)$, then we round to the right. On the other hand, if $(t+1)/m_i > t_i/m_i$, then $t=t_i$ (because $\varepsilon < 1/m_i$) and $A_{\ell_1}$ was taking at most $\varepsilon$ value from the block at $[(j-1)/m, j/m]$ (since $w_i(A_{\ell_1}) \leq t_i/m_i + \varepsilon$ before moving the cut). This means that by rounding to the left, we have added at most $\varepsilon$ value to $w_i(A_{\ell_2})$. Thus, we now have $w_i(A_{\ell_2}) \leq t_i/m_i + 2\varepsilon$.

We now proceed to round the second cut of the $i$-chain. Again, if we round to the left then $w_i(A_{\ell_2}) = t/m_i$ for some $t$, and if we round to the right $w_i(A_{\ell_2}) = (t+1)/m_i$. If $(t+1)/m_i \leq t_i/m_i$, then we round to the right. On the other hand, if $(t+1)/m_i > t_i/m_i$, then $t=t_i$ (because $2\varepsilon < 1/m_i$) and $A_{\ell_2}$ was taking at most $2\varepsilon$ value from the block that is cut by the $i$-cut. Thus, rounding to the left ensures that $w_i(A_{\ell_2}) = t_i/m_i$ and $w_i(A_{\ell_3}) \leq t_i/m_i + 3\varepsilon$. We keep repeating this until the rightmost $i$-cut of the $i$-chain has been rounded. At the end, we get that $w_i(A_{\ell_{k+1}}) \leq t_i/m_i + (k+1) \varepsilon$. However, we have $(k+1) \varepsilon < 1/m_i$ since $k+1 \leq n-1$ ($A_i$ cannot be one of the intervals). After we perform this procedure for all $i$-chains, agent $i$ will not envy any other agent.

We have shown a reduction to \textsc{unary-$\varepsilon$-EF-Cake-Cutting}. In order to obtain a reduction to \textsc{simple-$\varepsilon$-EF-Cake-Cutting}, we combine this reduction with Proposition \ref{prop:continuous-to-discrete}. Indeed, this yields a reduction from \textsc{Disjoint-Discrete-EF-Cake-Cutting} to \textsc{Disjoint-Discrete-EF-Cake-Cutting} where all agents value the same number of items. This means that we now have $m_i=m_j$ for all $i,j \in [n]$. By adding additional items that are not valued by anyone, we can also ensure that the number of items $m$ is a multiple of $m_i$. Applying the same reduction described in the first part of this proof to this instance yields an instance of \textsc{simple-$\varepsilon$-EF-Cake-Cutting} with $M=m$.
\end{proof}

\section{Conclusion}

In this paper, we study the classical cake cutting problem with the contiguity constraint and establish several hardness results and approximation algorithms for this setting.
It is worth noting that while our $1/3$-envy-free algorithm (Algorithm~\ref{alg:approx-EF-general}) is simple, lowering the envy to $1/4$ for the restricted class of uniform single-interval valuations (Algorithm~\ref{alg:approx-EF-single-interval}) already requires significantly more work.
Pushing the approximation factor down further even for this class or the class of piecewise uniform valuations while maintaining computational efficiency is therefore a challenging direction.
Of course, it is possible that there are hardness results for sufficiently small constants---this is not implied by the work of \citet{DengQiSa12}, as their PPAD-completeness result relies on more complex preference functions.

On the hardness front, we provide constructions that serve as frameworks for deriving NP-hardness results for both cake cutting and indivisible items.
Nevertheless, our frameworks do not cover questions related to the utilities of the agents, for instance whether there exists a contiguous envy-free allocation of the cake in which the first agent receives at least a certain level of utility.
Extending or modifying our constructions to deal with such questions is an intriguing direction for future research.

Finally, while we have established a number of connections between continuous and discrete cake-cutting in this paper, much still remains to be explored. For example, \citet{Suksompong19} showed in the discrete setting that if the valuations are binary, then an ``envy-free up to one item'' allocation is guaranteed to exist. Similarly, for additive (and even monotonic) valuations, \citet{BiloCaFl19} proved the existence of an allocation that is envy-free up to two items. It would be interesting to see how these problems can be related to the continuous setting.

\section*{Acknowledgments}

This work was partially supported by the European Research Council (ERC) under grant number 639945
(ACCORD) and by an EPSRC doctoral studentship (Reference 1892947).
We would like to thank the anonymous reviewers of the 34th AAAI Conference on Artificial Intelligence (AAAI 2020) for their valuable comments.

\bibliographystyle{named}
\bibliography{cake}

\end{document}